\documentclass{article}
\usepackage{qctl2qbf}

\usepackage{tikz}
\usepackage{config-tikz}
\usepackage{xspace}
\usepackage{slashed}
\usepackage{yhmath}
\usepackage{multirow}
\usepackage{amsmath}
\usepackage{amssymb}
\usepackage{stackengine}
\usepackage{scalerel}
\usepackage{url}

\newtheorem{theorem}{Theorem} 
\newtheorem{lemma}[theorem]{Lemma} 
\newtheorem{proposition}[theorem]{Proposition} 
 
\newtheorem{definition}{Definition} 
\newtheorem{example}{Example} 

\newenvironment{proof}{\paragraph{Proof:}}{\hfill$\square$}

\usepackage{boxedminipage}
\newenvironment{maboite}{\medskip\noindent\begin{boxedminipage}{\linewidth}}{\end{boxedminipage}\medskip}

\title{\QCTL\ model-checking with  \QBF \ solvers} 

\author{Akash Hossain \\ IRIF, Univ.\ Paris Diderot \and Fran\c{c}ois Laroussinie \\ IRIF, Univ.\ Paris Diderot  }

\begin{document}
\maketitle

\begin{abstract}
Quantified \CTL\ (\QCTL) extends the temporal logic \CTL\ with quantifications over atomic propositions. 
This extension is known to be very expressive: \QCTL\ allows us to express complex properties over Kripke structures (it is  as expressive as \MSO).
Several semantics exist for the quantifications: 
here, we work with  the \emph{structure semantics},  where  the extra propositions label 
the Kripke structure (and not its execution tree), and the model-checking problem is known to be \PSPACE-complete in this framework. 
We propose a new model-checking algorithm for \QCTL\  based on a reduction to \QBF. We consider several reduction strategies and we  compare them with a prototype (based on several \QBF\ solvers) on different examples. 
\end{abstract}



\section{Introduction}

Temporal logics have been introduced in computer science in the late
1970's by Pnueli~\cite{Pnu77}; they~provide a powerful formalism for specifying  correctness properties of
 evolving systems. Various kinds of temporal
logics have been defined, with different expressive power and
algorithmic properties. For~instance, the~\emph{Computation Tree Logic}~(\CTL)
expresses properties of the computation tree of the system under study (time is branching: a state may have several successors), and the 
\emph{Linear-time Temporal
Logic}~(\LTL) expresses properties of one execution at a time (a system is viewed as a set of executions).

Temporal logics allow \emph{model checking}, i.e.\ the automatic verification that a finite state system satisfies its expected behavioural specifications~\cite{QS82a,CE82}. It is well known that \CTL~model-checking is \PTIME-complete and \LTL~model-checking (based on automata techniques) is \PSPACE-complete. Verification tools exist for both logics and model-checking is now commonly used in the design of critical reactive systems. The main limitation to this approach is the state-explosion problem: symbolic techniques (for example with  BDD), SAT-based approaches, or partial order reductions have been developed and they are impressively successful. The SAT-based model-checking consists in using  SAT-solvers in the decision procedures. It was first developed for bounded model-checking (to search for executions whose length is bounded  by some integer,  satisfying some temporal property) which can be reduced to some satisfiability problem and then can be solved by a SAT-solver~\cite{BiereCCZ99}. SAT approaches have also been extended to unbounded verification and combined with other techniques~\cite{McMillan02}. Many studies have been done in this area, and it is widely considered as  an important approach in practice, which complements other symbolic techniques like BDD ones (see~\cite{BiereCCSZ03} for a survey).

In~terms of expressiveness, \CTL\ (or  \LTL) still has some limitations: in~particular,
it~lacks the ability of \emph{counting}. For instance, it~cannot express that
an event occurs (at~least) at every even position along a path, or that a
state has two successors. In~order to cope with this, temporal logics have
been extended with \emph{propositional quantifiers}~\cite{sis83}: those
quantifiers allow for adding fresh atomic propositions in the model before
evaluating the truth value of a temporal-logic formula. That~a state has at
least two successors can then be expressed (in~\emph{quantified \CTL},
hereafter written~\QCTL) by saying that it~is possible to label the model with
atomic proposition~$p$ in such a way that there is a successor that is labelled
with~$p$ and one that is~not. 

Different semantics for \QCTL\ have been studied in the literature depending on the definition of the labelling: either it refers to the finite-state model -- it is the \emph{structure} semantics -- or it refers to the execution tree -- it is the \emph{tree semantics}. 
Both semantics are interesting and have been extensively studied~\cite{Kup95a,Fre01,PBDDC02,Fre03,DLM12,LM14}.
While the tree semantics allow us to use the tree automata techniques to get decision procedures (model-checking and satisfiability are \TOWER-complete~\cite{LM14}), the situation is quite different for the structure semantics: in this framework, model-checking is \PSPACE-complete and satisfiability is undecidable~\cite{Fre01}.

In this paper, we focus  on the structure semantics. We first motivate this choice by showing that \QCTL\ may encode many  logics, for example we explain how to reduce model-checking for \SML\ (Sabotage Modal Logic)~\cite{Benthem05} to the \QCTL\ model-checking problem. Then we propose a model-checking algorithm based on a reduction to \QBF\ (propositional logic augmented with quantifiers): given a Kripke structure $\calK$ and a \QCTL\ formula $\Phi$, we show how to build a \QBF\ formula $\widehat{\Phi}^\calK$ which is valid iff $\calK\sat \Phi$.  It is  natural to use \QBF\  quantifiers  to deal with propositional  quantifiers of \QCTL. Of course, \QBF-solvers are not as efficient as SAT-solvers, but still much progress has been made   and \QBF-solvers have already been considered for model-checking, as in~\cite{DershowitzHK05,qfm}.
Here we propose several reductions depending on the way of dealing with nested temporal modalities, and we compare them with a prototype we implemented (connected to different solvers: \Z3~\cite{MouraB08}, \qfm~\cite{qfm}, \cqesto~\cite{cqesto} and  \qfun~\cite{qfun}). As far as we know, it is the first implementation of a model-checker for \QCTL. 

Here, our first objective is  to use the \QBF-solver as a tool to check complex properties over limited size models, and this is therefore different from the classical use of  \SAT-based techniques which are precisely applied to solve verification problems for very large systems.  

The outline of the paper is as follows: we~begin with setting up the necessary
formalism in order to define \QCTL\ and to discuss its semantics. We~then devote
Section~\ref{sec-reduc} to the different reductions to \QBF. Finally,
Section~\ref{sec-experiments} contains several practical results and examples.


\section{Definitions}

\subsection{Kripke structures}

Let $\AP$ be a finite set of atomic propositions.

\begin{definition}
A~\emph{Kripke structure} is a tuple $\calK=\tuple{V,E,\ell}$, where $V$~is a
finite set of vertices (or states), $E\subseteq V\times V$ is a set of edges
(we assume  that for any~$x\in V$, there exists $x'\in V$ s.t. $(x,x')\in E$),
and $\ell\colon V\to 2^{\AP}$ is a labelling function.
\end{definition}

An infinite  path (also called an execution) in a Kripke structure is an infinite  sequence $\rho = x_0 x_1 x_2 \ldots$  such that for any $i$ we have $x_i \in V$ and $(x_i,x_{i+1})\in E$. We~write  $\Path^\omega_\calK$ for the
set of infinite paths of~$\calK$ and $\Path^\omega_\calK(x)$ for the set of infinite paths issued from $x\in V$.
Given such a path $\rho$, we use $\rho_{\leq i}$ to denote the $i$-th prefix $x_0\ldots x_i$, $\rho_{\geq i}$ for the $i$-th suffix $x_i x_{i+1}\ldots$, and $\rho(i)$ for the vertex $x_i$. The size of $\calK$ is $|V|+|E|$.

 Given a set $P\subseteq \AP$, two  Kripke structures $\calK=(V,E,\ell)$ and $\calK'=(V',E',\ell')$ are said $P$-equivalent (denoted by $\calK \equiv_P \calK'$) if $V=V'$, $E=E'$, and for every $x \in V$ we have: $\ell(x)\cap P = \ell'(x) \cap P$.

\subsection{\QCTL}

This section is devoted to the definition of the logic \QCTL ,\ which extends the classical branching-time temporal logic \CTL\ with  quantifications over atomic propositions.   

\begin{definition}
\label{def-QCTL}
The syntax of $\QCTL$ is defined by the following grammar:
\begin{xalignat*}1
\QCTL \ni \phi ,\psi  &\coloncolonequals  q \mid \neg\phi \mid \phi \ou\psi  
    \mid \Ex \X \phi  \mid \Ex \phi \Until \psi \mid \All \phi \Until \psi  \mid \exists p.\ \phi 
\end{xalignat*}
where $q$ and~$p$ range over~$\AP$.
\end{definition}

\QCTL\  formulas are evaluated over states of  Kripke structures:

\begin{definition}
Let $\calK=\tuple{V,E,\ell}$ be a Kripke structure, and $x \in V$. 
The semantics of \QCTL\ formulas is defined inductively as follows: 
\begin{align*}
\calK, x \models & p  \text{ iff }  p\in \ell(x) \\
\calK,  x \models & \neg\phi  \text{ iff }  \calK,  x \not\models \phi \\
\calK,  x \models & \phi \ou \psi  \text{ iff }  \calK,  x
\models \phi \text{ or }  \calK,  x \models \psi \\
\noalign{\pagebreak[2]}
\calK,  x \models & \Ex \X \phi   \text{ iff }  \exists (x,x') \in E \text{ s.t.\ } \calK, x' \models \phi \\
\calK,  x \models & \Ex \phi  \Until \psi  \text{ iff }  \exists  \rho  \in \Path^\omega_\calK(x),  
  \exists i\geq 0  \text{ s.t.\ }  \calK,  \rho(i) \models \psi \text{ and } \\
   & \qquad \qquad \qquad \qquad \qquad \text{ for any }  0 \leq j <i, \text{ we have }  \calK, \rho(j) \models \phi  \\
\calK,  x \models &\All \phi  \Until \psi  \text{ iff }  \forall  \rho  \in \Path^\omega_\calK(x),  
  \exists i\geq 0  \text{ s.t.\ }  \calK,  \rho(i) \models \psi \text{ and }  \\
  & \qquad \qquad \qquad  \qquad\qquad  \text{ for any }  0 \leq j <i, \text{ we have }  \calK, \rho(j) \models \phi  \\
  \calK,  x   \models & \exists p.\ \phi  \text{ iff } \exists \calK'
\equiv_{\AP\backslash\{p\}} \calK  \text{ s.t. } \calK', x \models \phi 
\end{align*}
\end{definition}

In~the sequel, we use standard abbreviations such as $\top$, $\bot$, $\et$, $\impl$ and $\equivalent$.
We also use the additional    temporal modalities  of \CTL: $\All\X\phi \eqdef \non \Ex \X \non \phi$
, $\Ex\F\phi \eqdef \Ex \top \Until \phi$, $\All\F\phi \eqdef \All \top \Until \phi$, 
$\Ex \G\phi \eqdef \neg\All \F\neg\phi$,  $\All \G\phi \eqdef \neg\Ex \F\neg\phi$,  $\Ex \vfi \WUntil \psi \eqdef \non \All \non \psi \Until (\non \psi \et \non \vfi)$ and $\All \vfi \WUntil \psi \eqdef \non \Ex \non \psi \Until (\non \psi \et \non \vfi)$. 

Moreover, we use the following abbreviations related to quantifiers over atomic  propositions:   $\forall p.\ \phi \eqdef \neg\exists p.\ \neg\phi$,  and for a set $P = \{p_1,\ldots, p_k\} \subseteq \AP$, we write $\exists P. \phi$ for  $\exists p_1.  \ldots \exists p_k. \phi$ and $\forall P. \phi$ for  $\forall p_1. \ldots \forall p_k. \phi$.

 The~\emph{size} of a formula~$\phi\in\QCTL$, denoted $\size\phi$, is defined inductively by: $\size q = 1$, $\size{ \non \phi} = \size{\exists p.\phi} = \size{\Ex\X \phi} =  1 + \size \phi$, $\size{\phi \ou \psi} = \size{\Ex \phi \Until \psi}= \size{\All \phi \Until \psi}= 1 + \size \phi + \size \psi$.  
The \emph{temporal height} of  $\vfi$, denoted $\HT(\vfi)$,   is the maximum number of nested temporal modalities in $\vfi$: $\HT(q)=0$, $\HT(\phi \et \psi)=\HT(\phi \ou \psi)=\max(\HT(\phi),\HT(\psi))$,
 $\HT(\non \phi)=\HT(\exists p.\phi)=\HT(\phi)$, $\HT(\EX \phi)=1+\HT(\phi)$, and $\HT(\Ex \phi \Until \psi)=\HT(\All \phi \Until \psi)= 1+\max(\HT(\phi),\HT(\psi))$. 
 And given a subformula $\psi$ in $\Phi$, the \emph{temporal depth} of $\psi$ in $\Phi$ (denoted $\tdepth_\Phi(\psi)$)  is the number of temporal modalities having $\psi$ in their scope in $\Phi$.

Two \QCTL\ formulas $\vfi$ and $\psi$ are said to be \emph{equivalent} (written $\vfi \equiv \psi$) iff for any structure $\calK$, any state $x$, we have $\calK, x \sat \vfi$ iff   $\calK, x \sat \psi$. This equivalence is substitutive~\footnote{If $\phi\equiv\psi$, replacing the subformula $\phi$ by $\psi$ in a formula $\Phi$ does not change the truth value of $\Phi$.}.

A formula $\Phi$ is said to be in \emph{Negation Normal Form} (NNF) when the negations are only applied to atomic propositions in $\Phi$: any \QCTL \ formula is equivalent to some formula in NNF built from operators in $\{\et,\ou,\exists,\forall, \EX, \AX, \Ex\_\Until,$  $\All\_\Until, \Ex\_\WUntil, \All\_\WUntil\}$ and literals $p$ and $\non p$ with $p\in \AP$. 

\subsection{Discussion on the semantics.}

The semantics we defined is classically called the \emph{structure semantics} (or \emph{Kripke semantics} in~\cite{Fre01}): a formula $\exists p.\phi$ holds true in a Kripke structure $\calK$ iff there exists a $p$-labelling of the structure $\calK$ such that $\phi$ is satisfied. 
 Another  well-known semantics coexists in the literature for propositional
quantifiers,  the  \emph{tree semantics}: $\exists p.\ \phi$ holds
 true when there exists a labelling by $p$ of the \emph{execution tree} (the infinite unfolding) of the
 Kripke structure under which $\phi$~holds. If, for \CTL, interpreting formulas over the structure or the execution tree is equivalent, this is not the case for \QCTL. For example, $\forall p. (p \impl \EX p)$ specifies the existence of a self-loop in the current state when interpreted in the structure semantics, and it is never true in the tree semantics. 
Finally note that there is also 
 the \emph{amorphous semantics}~\cite{Fre01}, where
$\exists p.\ \phi$ holds true at a state~$x$ in some Kripke structure~$\calK$
if, and only if, there exists some Kripke structure~$\calK'$ with a state $x'$
such that $x$ and $x'$ are bisimilar, and for which there exists a $p$-labelling
making $\phi$ hold true at~$x'$. With this last semantics, the logic is insensitive
to unwinding, and more generally it is bisimulation-invariant (contrary to the two previous semantics). 
 Now we compare the tree and the structure semantics. 
 
 \paragraph{Complexity} First  note that    these two semantics do not have the same algorithmic  properties: if \QCTL\ model-checking and satisfiability are \TOWER-complete for the tree semantics (the algorithms are based on tree automata techniques), \QCTL\ model-checking  is \PSPACE-complete for the structure semantics but satisfiability is undecidable (see~\cite{LM14} for a survey). 
   
\paragraph{Expressive power} 
In both semantics, \QCTL\ is as expressive\footnote{This would require adequate definitions,
since a temporal logic formula may only deal with the reachable part of the
model, while \MSO\ has a more \emph{global} point of view.} as the Monadic Second-Order Logic over the finite structures or the infinite trees (depending on the semantics) and as \QCTLs\ (the extension of the standard $\CTLs$ with state formulas $\exists p.\phi$). 
Note also that any \QCTL\ formula is equivalent to a formula in prenex normal form (we will use this result in next sections). All these results are presented in~\cite{LM14}.

\subsection{Motivations and examples for \QCTL\ in the structure semantics}

First we present several  examples of formulas to illustrate the expressive power of \QCTL. Then we will consider a more complex problem showing how  \QCTL\ can be useful to encode complex specifications written in other logics. 
 
\subsubsection{Examples of \QCTL\ formulas}

\QCTL\ allows us to express complex properties over Kripke structures: for example, we can build a characteristic formula (up to isomorphism) of a structure or  reduce model-checking problems for multi-player games to \QCTL\ model-checking~\cite{LaroussinieM15}. Below, we give several examples of counting properties, to illustrate the expressive power of propositional quantifiers.

The first   formula below expresses that there exists a unique reachable state satisfying $\vfi$, and the second one states that there exists a unique immediate successor satisfying $\vfi$:
\begin{align}
\Ex_{=1} \F\vfi &\;\eqdef\; \EF \vfi \et \forall p.\ \big(\EF (p\et \vfi) \impl \AG (\vfi \impl p)\big)  \label{FuniqF} \\
\Ex_{=1}\X\phi  &\;\eqdef\; \Ex\X\phi \et \forall p.\ \big(\Ex\X(\phi \et p) \impl \All\X(\phi \impl  p)\big) \label{FuniqX}
\end{align}
where we assume that $p$ does not appear in~$\vfi$. Consider the formula~(\ref{FuniqF}):  if there were two reachable states satisfying $\vfi$, then labelling only one of them with~$p$ would falsify the $\AG$ subformula. For (\ref{FuniqX}), the argument is similar. 

The existence of 
at least $k$ successors satisfying a given property can be expressed with:
\begin{align}
\Ex_{\geq k}\X\phi = & \exists p_1 \ldots \exists p_k.\ 
\Bigl(\ET_{1\leq i\leq k} \Ex\X \bigl(p_i\et \ET_{i'\not= i} \non p_{i'}\bigr) \et
\All\X\Bigl(\bigl(\OU_{1\leq i\leq k} p_i\bigr) \impl \phi\Bigr)\Bigr) 
\end{align}

And we can define $\Ex_{=k} \X \phi$ as $\Ex_{\geq k} \X \phi  \:\et\: \non \Ex_{\geq k+1} \X \phi$. Note that
these examples show why \QCTL\ formulas are not  bisimulation-invariant. 

When using \QCTL\ to specify properties, one often needs to quantify (existentially or universally) over \emph{one} reachable state we want to mark with a given atomic proposition. To this aim, we add the following abbreviations: 
\[
\exists^1p.\vfi \eqdef  \exists p. \big( (\Ex_{=1}\F \: p) \et \vfi\big) \quad\quad\quad
\forall^1p.\vfi \eqdef  \forall p. \big( (\Ex_{=1}\F \:  p) \impl \vfi\big)
\]

\subsubsection{From Sabotage Modal Logic to \QCTL}

\QCTL\ can be used to encode (quite easily) many problems for other temporal or modal logics. For example, in~\cite{LMS15}, \QCTL\ is used to decide problems for multi-agent systems (the tree and structure semantics are both used according to the type of strategies allowed in the system). Here we consider another example to motivate the use of \QCTL\ with the structure semantics: the model-checking problem for the Sabotage Modal Logic (\SML). 

\SML\ is a modal logic containing modalities which may \emph{delete} transitions in the model~\cite{Benthem05,AFH15,LR03,LR03b,ABG18}. \SML\ formulas are built from Boolean connectives, atomic propositions and the following modalities: $\Wdiam$ (and its dual $\Wbox$) and $\Bdiam$ (and its dual $\Bbox$). 
$\Wdiam$ is equivalent to the \CTL\ modality $\EX$ and $\Bdiam$ is a edge removal operator.
Given   a  Kripke structure $\calK=\tuple{V,E,\ell}$   and a state $x\in V$, the semantics of these modalities is as follows:
\begin{align*}
\tuple{V,E,\ell}, x \sat \Wdiam \vfi  & \quad \mbox{iff} \quad \exists (x,x') \in E \; \mbox{s.t.}\: \tuple{V,E,\ell}, x' \sat  \vfi \\
\tuple{V,E,\ell}, x \sat \Bdiam \vfi  & \quad \mbox{iff} \quad \exists (y,y') \in E \; \mbox{s.t.}\: \tuple{V,E\setminus \{(y,y')\},\ell}, x \sat  \vfi \\
\end{align*}
The expressive power of 
\SML\ is  interesting, one can express properties over the \emph{frame} of the underlying Kripke structure.  For example, the formula 
$\Wdiam\top \et \Wbox\Wdiam \top \et \Bbox \Wbox \bot$ 
holds true in   $\calK, x$ iff  the structure $\calK$ is restricted to a single selfloop from $x$. Other examples to characterize structures (\eg\: cycles of length $n$) are given in~\cite{ABG18}.
We know that \SML\ satisfiability is undecidable and \SML\ model-checking is \PSPACE-complete\cite{AFH15,LR03,LR03b}.
  
A local variant is also used in literature where the removed transition has to be issued from the current state. This is done with the modality $\Bdiaml$ defined by:
\[
\tuple{V,E,\ell}, x \sat \Bdiaml \vfi   \quad \mbox{iff} \quad \exists (x,x') \in E \; \mbox{s.t.}\: \tuple{V,E\setminus \{(x,x')\},\ell}, x \sat  \vfi 
\]

We can easily reduce a model-checking instance for \SML\ (including the local modality):  $\calK, x \sat \Phi$ to a model-checking instance for \QCTL\ $\calK', x' \sat \Phi'$. 
The main idea of the reduction consists in identifying $\Bdiam$-removed edges  by labelling  intermediary states along them  with some fresh atomic propositions $\mathsf{del}_i$. If $k$ is the $\{\Bdiam,\Bdiaml\}$-height of $\Phi$~\footnote{\ie\ the maximal number of nested  $\{\Bdiam,\Bdiaml\}$-modalities in $\Phi$.},    we will use the  propositions $\mathsf{del}_0$,\ldots, $\mathsf{del}_{k-1}$ in the formula $\Phi'$. 
Moreover,  we can see that any modality $\Bdiam$ allows us to remove any edge in $\calK$ (with no relationship with the current state where the formula is interpreted). To encode this, we need to complete the Kripke structure with additional edges connecting any pair of states (including selfloops). Formally $\calK'$ is then defined by:
\begin{itemize}
\item $V' \eqdef V \cup \{v_{xy} \:|\: (x,y) \in E\}$,
\item $E'\eqdef \{ (x,v_{xy}),(v_{xy},y) \:|\: (x,y)\in E\} \cup \{(x,y) \:|\: x,y \in V\}$
\item $\ell'(x) = \ell(x)$ for $x \in V$ and $\ell(v_{xy})=\{\mathsf{inter}\}$ for every $v_{xy}$ in $V'$. 
\end{itemize}
where $\mathsf{inter}$ is a fresh atomic proposition used to mark intermediary states along initial $\calK$'s edges.

Formally we define $\widetilde{\:\phi\:}^{\:n}$ where $n$ is the $\{\Bdiam,\Bdiaml\}$-depth of subformula $\phi$ in the main formula $\Phi$ (\ie\ $\phi$ occurs in the scope of $n$ nested modalities $\{\Bdiam,\Bdiaml\}$ in $\Phi$). The definition is given in Table~\ref{tab-sml} and the correctness of the reduction is stated as follows:
\begin{table}
\begin{align*}
\widetilde{\:p\:}^{\:n} & \eqdef p \quad \quad 
{\widetilde{\vfi \et \psi}}^{n}  \eqdef  {\widetilde{\:\vfi\:}}^{\:n} \et {\widetilde{\:\psi\:}}^{\:n}  \quad \quad
{\widetilde{\non  \phi}}^{\:n} \eqdef   \non \: {\widetilde{\:\phi\:}}^{\:n}   \quad \quad \widetilde{\:\top\:}^{\:n} \eqdef \top \\
{\widetilde{\Wdiam  \vfi}}^{\:n} & \eqdef   \EX (\mathsf{inter} \et \ET_{0 \leq i <n} \non \mathsf{del}_i \et  \EX \: {\widetilde{\:\vfi\:}}^{\:n})  \\
{\widetilde{\Bdiam  \vfi}}^{\:n} & \eqdef   \exists^1 \mathsf{del}_{n} . \Big( \EX \: \EX (\mathsf{inter} \et \ET_{0 \leq i <n} \!\!\non \mathsf{del}_i \et \mathsf{del}_n) \et  {\widetilde{\:\vfi\:}}^{\:n+1}\Big)  \\
{\widetilde{\Bdiaml  \vfi}}^{\:n} & \eqdef   \exists^1 \mathsf{del}_{n} . \Big(\EX (\mathsf{inter} \et \ET_{0 \leq i <n} \!\! \non \mathsf{del}_i \et \mathsf{del}_n) \et  {\widetilde{\:\vfi\:}}^{\:n+1}  \Big)
\end{align*}
\caption{Transformation rules from \SML\ to \QCTL}
\label{tab-sml}
\end{table}
\begin{proposition}
Let $\calK$ be a Kripke structure $\tuple{V,E,\ell}$, $x \in V$,  and $\Phi\in \SML$. Given a $\Phi$-subformula $\psi$ such that $\psi$ occurs at   $\{\Bdiam,\Bdiaml\}$-depth $n$ in $\Phi$, and given $E_d \subseteq E$ with 
$E_d=\left\lbrace (x_0, y_0), \ldots (x_{n-1}, y_{n-1})\right\rbrace$, we have:
\[
\tuple{V,E\setminus E_d,\ell},x \sat \psi \quad\quad \mbox{iff}\quad \quad \tuple{V',E',\ell''}, x \sat \widetilde{\:\psi\:}^{\:n}
\]
where $\calK'=\tuple{V',E',\ell'}$ is defined as above, and $\ell''(v)$ for $v \in V'$ is defined as follows:
if $v=v_{x_iy_i}$ for some $i\in \{0,\ldots,n-1\}$, then $l''(v)=\{\mathsf{inter},\mathsf{del}_i\}$, else $l''(v)=l'(v)$.
\end{proposition}

\begin{proof} 
The proof is done by structural induction on $\psi$. We only consider the modalities $\Wdiam$ and $\Bdiam$:
\begin{itemize}
\item $\psi = \Wdiam \psi_1$: If $\tuple{V,E\setminus E_d,\ell},x \sat \psi$, there exists $(x,y)\in E\setminus E_d$ such that 
$\tuple{V,E\setminus E_d,\ell},y \sat \psi_1$. By i.h., we get $\tuple{V',E',\ell''}, y \sat \widetilde{\psi_1}^{\:n}$, from which we deduce $\tuple{V',E',\ell''}, x \sat \widetilde{\:\psi\:}^{\:n}$ because the intermediary state $v_{xy}$ is labelled by $\ell'(v_{xy})=\{\mathsf{inter}\}$ (\ie\ no  $\mathsf{del}_i$ is true at $v_{xy}$). 
The other direction proceeds in the same way.
\item $\psi = \Bdiam \psi_1$: If $\tuple{V,E\setminus E_d,\ell},x \sat \psi$, there exists $(y,y')\in E\setminus E_d$ such that 
$\tuple{V,E\setminus (E_d\cup\{(y,y')\}),\ell},x \sat \psi_1$. By i.h., we get $\tuple{V',E',\ell''}, x \sat \widetilde{\psi_1}^{\:n+1}$ with a labelling $\ell''$ as described in the proposition, in particular we have $\ell''(v_{yy'})=\{\mathsf{inter},\mathsf{del}_n\}$. We can deduce that we have $\tuple{V',E',\ell'''}, x \sat \widetilde{\:\psi\:}^{\:n}$ where $\ell'''$ coincides with $\ell''$ except for $v_{yy'}$ that is labelled only by $\mathsf{inter}$. The other direction is similar.
\end{itemize}
\hfill$\Box$
\end{proof}
  
In particular, we have $\calK,x\sat \Phi$ iff $\calK',x\sat \widetilde{\:\Phi\:}^0$ which provides the reduction.  
  
\begin{example}
To illustrate the construction of $\calK'$, consider the structure $\calK$ in Figure \ref{fig-ex-sml} and its corresponding $\calK'$. And  $\calK,x_1 \sat \Wdiam \Wdiam \Bbox \Wdiam \top$  is reduced to:
$\calK',x_1 \sat \EX \Big(\mathsf{inter} \et \EX (\EX (\mathsf{inter} \et \EX (\non \exists^1 \mathsf{del}_0. \EX \EX (\mathsf{inter} \et \mathsf{del}_0) \et \non \EX (\mathsf{inter} \et \non\mathsf{del}_0 \et \EX \top)))) \Big)$. Note that the formula holds true because $x_3$ has 2 successors. 
  
\begin{figure}[h]
\centering
\begin{tikzpicture}
\draw (0,3) node (K1) {$\calK$} ;
\draw (0,2) node[draw,circle,inner sep=0.1mm,vert] (x1) {$x_1$} ;
\draw (2,2) node[draw,circle,inner sep=0.1mm,vert] (x3) {$x_3$} ;
\draw (1,3) node[draw,circle,inner sep=0.1mm,vert] (x2) {$x_2$} ;
\draw (2,0) node[draw,circle,inner sep=0.1mm,vert] (x4) {$x_4$} ;
\draw (0,0) node[draw,circle,inner sep=0.1mm,vert] (x5) {$x_5$} ;

\draw[-latex'](x1) -- (x2) ;
\draw[-latex'](x2) -- (x3) ;
\draw[-latex'](x3) -- (x1) ;
\draw[-latex'](x5) -- (x4) ;
\draw[out=-125,in=125,-latex'] (x3) to  (x4);
\draw[out=55,in=-55,-latex'] (x4) to (x3);

\begin{scope}[shift={(5,0)}]
\draw (-0.5,3) node (K2) {$\calK'$} ;
\draw (0,2) node[draw,circle,inner sep=0.1mm,vert] (x1) {$x_1$} ;
\draw (1,2) node[draw,circle,inner sep=0.1mm,vert] (x31) {$\;$} ;
\draw (2,2) node[draw,circle,inner sep=0.1mm,vert] (x3) {$x_3$} ;
\draw (0.5,2.5) node[draw,circle,inner sep=0.1mm,vert] (x12) {$\;$} ;
\draw (1.5,2.5) node[draw,circle,inner sep=0.1mm,vert] (x23) {$\;$} ;

\draw (1,3) node[draw,circle,inner sep=0.1mm,vert] (x2) {$x_2$} ;
\draw (1,0) node[draw,circle,inner sep=0.1mm,vert] (x54) {$\;$} ;
\draw (1.7,1) node[draw,circle,inner sep=0.1mm,vert] (x34) {$\;$} ;
\draw (2.3,1) node[draw,circle,inner sep=0.1mm,vert] (x43) {$\;$} ;

\draw (2,0) node[draw,circle,inner sep=0.1mm,vert] (x4) {$x_4$} ;
\draw (0,0) node[draw,circle,inner sep=0.1mm,vert] (x5) {$x_5$} ;

\draw[-latex'](x1) -- (x12) ;
\draw[-latex'](x12) -- (x2) ;
\draw[-latex'](x2) -- (x23) ;
\draw[-latex'](x23) -- (x3) ;
\draw[-latex'](x3) -- (x31) ;
\draw[-latex'](x31) -- (x1) ;
\draw[-latex'](x5) -- (x54) ;
\draw[-latex'](x54) -- (x4) ;
\draw[-latex'](x3) -- (x34) ;
\draw[-latex'](x34) -- (x4) ;
\draw[-latex'](x4) -- (x43) ;
\draw[-latex'](x43) -- (x3) ;

\draw[latex'-latex',dashed, line width=0.5](x1) -- (x4) ;
\draw[latex'-latex',dashed, line width=0.5](x1) -- (x5) ;
\draw[latex'-latex',dashed, line width=0.5](x3) -- (x5) ;
\draw[out=150,in=180,latex'-latex',dashed,line width=0.5,looseness=1.5] (x5) to  (x2);
\draw[out=30,in=0,latex'-latex',dashed,line width=0.5,looseness=1.5] (x4) to  (x2);

\draw[out=70,in=-160,latex'-latex',dashed,line width=0.5] (x1) to  (x2);
\draw[out=110,in=-20,latex'-latex',dashed,line width=0.5] (x3) to  (x2);
\draw[out=-25,in=-155,latex'-latex',dashed,line width=0.5] (x1) to  (x3);
\draw[out=-45,in=45,latex'-latex',dashed,line width=0.5] (x3) to  (x4);
\draw[out=-25,in=-155,latex'-latex',dashed,line width=0.5] (x5) to  (x4);
\draw[loop left,-latex',dashed,line width=0.5] (x1) to (x1);
\draw[loop left,-latex',dashed,line width=0.5] (x5) to (x5);
\draw[loop above,-latex',dashed,line width=0.5] (x2) to (x2);
\draw[loop right,-latex',dashed,line width=0.5] (x3) to (x3);
\draw[loop right,-latex',dashed,line width=0.5] (x4) to (x4);
\draw[loop left,-latex',dashed,line width=0.5] (x5) to (x5);

\end{scope}

\end{tikzpicture}
\caption{Example of reduction for \SML\ model-checking.}
\label{fig-ex-sml}
\end{figure}
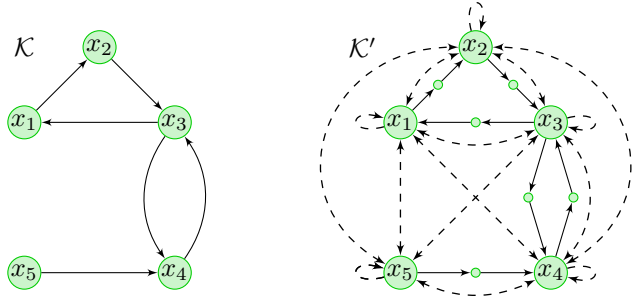
  \end{example}


\section{Model-checking QCTL}
\label{sec-reduc}

Model-checking \QCTL\ is a \PSPACE-complete problem, and it is \NP-complete for the  restricted set of formulas of the form $\exists P.\vfi$, with $P\subseteq \AP$ and $\vfi \in \CTL$ \cite{Kup95a}.
In this section, we give a reduction from the \QCTL\ model-checking problem to the \QBF \ validity problem.

In the following, we assume a Kripke structure $\calK=\tuple{V,E,\ell}$ with $V=\{x_0,\ldots,x_n\}$, an initial state $x_0 \in V$ and a \QCTL\ formula $\Phi$ to be fixed.  
We also assume w.l.o.g.\  that every quantifier $\exists$ and $\forall$ in $\Phi$ introduces a fresh atomic proposition, and distinct from the propositions used in $\calK$. We use $\AP_Q^\Phi$ to denote the set of  quantified atomic propositions in $\Phi$. 

These  assumptions allow us to use an alternative notation for the semantics of $\Phi$-subformulas: the truth value of $\vfi$ will be defined for a state $x$ in  $\calK$ within an \emph{environment} $\varepsilon:\AP_Q^\Phi\rightarrow 2^V$, that is a partial mapping  associating a subset of vertices to a proposition in $\AP_Q^\Phi$. We use $\calK,x \sat_\varepsilon \vfi$ to denote that $\vfi$ holds at $x$ in $\calK$ within $\varepsilon$.
Therefore  the $\calK$'s labelling $\ell$ is not modified when a subformula is evaluated, only $\varepsilon$ is extended with labellings for  new quantified propositions. 
Formally the main changes of the semantics are as follows:
\begin{align*}
\calK, x \models_\varepsilon p \text{ iff } & \Big((p \in \AP_Q^\Phi \text{ and } x \in \varepsilon(p)) \text{ or } (p \not\in \AP_Q^\Phi \text{ and } p \in \ell(x))\Big) \\
\calK,  x \models_\varepsilon \exists p.\ \phi \text{ iff } & \exists V' \subseteq V  \text{ s.t. } \calK, x \models_{\varepsilon[p \mapsto V']} \phi 
\end{align*}
where $\varepsilon[p \mapsto V']$ denotes the mapping    which coincides with $\varepsilon$ for every proposition in $AP_Q^\Phi\setminus \{p\}$ and associates $V'$ to $p$. 

We use this new notation in order to better distinguish  initial $\calK$'s propositions and quantified propositions to make proofs simpler. Of course, there is no semantic difference: $\calK,x  \models \Phi$ iff $\calK, x \models_\emptyset \Phi$. 

\medskip

In next sections, we consider general \emph{quantified propositional formulas} (\QBF) of the form:
\[
\QBF \ni \alpha, \beta \coloncolonequals  q \mid   \alpha \ou \beta \mid  \non \alpha \mid    \exists q.\alpha \mid  
\]
We will also use the following classical abbreviations: 
$\alpha \et \beta = \non (\non \alpha \ou \non \beta)$, $\alpha \impl \beta = \non \alpha \ou \beta$, 
$\alpha \equivaut \beta = (\alpha \impl \beta) \et (\beta \impl \alpha)$, and $\forall q.\alpha = \non\exists q. \non \alpha$. 
The formal semantics of  a formula $\alpha$ is defined over a Boolean valuation for free variables in $\alpha$ (\ie\ propositions which are not bound by a quantifier~\footnote{We assume w.l.o.g.\ that every quantifier $\exists$ or $\forall$ introduces a new proposition.}), and it is defined as usual. A  formula  is said to be closed when it does not contain free variables. In the following, we use the standard notion of validity for closed \QBF\ formulas.

Our aim is then to build a (closed) \QBF\ formula $\widehat{\Phi}^{x_0}$ such that   $\widehat{\Phi}^{x_0}$  is valid iff $\Phi$ holds true at $x_0$ in $\calK$. 

\subsection{Overview}

We present several reductions from the \QCTL\ model-checking problem to the \QBF\ validity problem. 
These reductions are defined as two steps processes: first a pre-processing is carried out  on  the original \QCTL\ formula and then a syntactic translation into \QBF\ is applied. 
All these reductions differ only from the pre-processing step, indeed they  share the   final translation denoted  $\widehat{\vfi}^{x,P}$ for a  \QCTL\  formula $\Phi$, a vertex $x$ and a subset  $P \subseteq AP_Q^\Phi$, and defined in Table~\ref{reduc-main} (its correctness will be established in Theorem~\ref{theo-UU}). 

Of course, the construction of  the \QBF \ formula $\widehat{\vfi}^{x,P}$ 
uses the structure $\calK$ (the transition relation $E$ and the labelling $\ell$). The vertex $x$ is the state where the \QCTL\ formula $\phi$ has to be interpreted. Every quantification $\exists p$ in $\Phi$ is replaced by a sequence of quantifications $\exists p^{x_0} \ldots \exists p^{x_n}$ in the \QBF\ formula in order to encode the $p$-labelling of $\calK$ (\ie\ a truth value for every state of the structure): the variable $p^{x_i}$ is assumed to be true iff $x_i$ is labelled by $p$. The set $P$ in  $\widehat{\vfi}^{x,P}$ is the set of quantified propositions: when evaluating a proposition $p$ at $x$, we need to know whether $p$ belongs to the set of initial atomic propositions of $\calK$ (and its truth value depends on $\ell(x)$), or $p$ is a quantified proposition introduced by some quantifier $\exists$ or $\forall$ in $\Phi$ (and its truth value is the variable $p^x$). The encoding of the temporal modalities are explained below.

\subsubsection{Unfolding characterization of the until operators}

In Table~\ref{reduc-main}, the temporal modalities are encoded in \QBF \ by unfolding of the transition relation $E$.
For example,  $\EX \phi$ holds true at $x$ iff there exists some $(x,x')\in E$ such that $x' \sat \phi$, this is precisely the meaning of the corresponding rule in the table with the disjunction over the $x'$s. And the truth value of $\EF \phi$ at $x$ is encoded as a disjunction of the truth values of $\phi$ at any state $x'$  reachable from $x$ with an arbitrary number of transitions in $E$ ($E^*$ denotes the reflexive and transitive closure of $E$). 
The rules for $\AX$ and $\AG$ are similar. 
Finally the rules for $\EU$ and $\AU$ can be seen as a depth-first way to look for a  path (or a set of paths for $\AU$) satisfying the Until modality. Other temporal modalities (\eg\ $\EG$, $\AF$, $\EW$, or $\AW$) can be translated into \QBF \ by using their definitions in terms of $\EU$ and $\AU$ or by using adhoc rules as for $\AG$~\footnote{The presence of dedicated rules for $\EF$, $\AG$ and $\AX$ is due to the fact that other reductions eliminate all temporal modalities except $\EX$, $\AX$, $\EF$ and $\AG$.  Of course the formulas provided by these rules are equivalent -- modulo Boolean simplifications -- to the formulas we could obtain by using the standard definitions $\EF = \Ex \top \Until \_$, $\AG = \non \EF \non$ and $\AX = \non \EX \non$.}.

\begin{table}[ht]
\begin{align*}
\widehat{\non \vfi}^{\:x,P} &\eqdef \non \widehat{\vfi}^{\:x,P} \quad \quad 
\widehat{\vfi \ou \psi}^{\:x,P} \eqdef  \widehat{\vfi}^{\:x,P} \ou \widehat{\psi}^{\:x,P} \quad \quad 
\widehat{\vfi \et \psi}^{\:x,P} \eqdef  \widehat{\vfi}^{\:x,P} \et \widehat{\psi}^{\:x,P}  \\
\widehat{\exists p. \vfi}^{\:x,P} &\eqdef \exists p^{x_0} \ldots p^{x_n}.  \widehat{\vfi}^{\:x,P\cup\{p\}} \quad\quad\quad
\widehat{p}^{\:x,P} \eqdef \begin{cases} p^x & \text{if}\: p \in P \\ \top & \text{if}\:  p \not\in P \text{ and } p \in \ell(x) \\ \bot & \text{otherwise} \\ \end{cases} \\ 
\widehat{\EX  \vfi}^{\:x,P} & \eqdef  \!\!\!\! \OU_{(x,x') \in E} \!\!\!\!\widehat{\vfi}^{\:x',P} \quad \quad \quad
\widehat{\AX  \vfi}^{\:x,P} \eqdef  \!\! \!\!\ET_{(x,x') \in E} \!\!\!\!\widehat{\vfi}^{\:x',P}  \\
\widehat{\EF  \vfi}^{\:x,P} & \eqdef \!\!\!\!  \OU_{(x,x') \in E^*}\!\!\!\! \widehat{\vfi}^{\:x',P} \quad \quad  \quad
\widehat{\AG  \vfi}^{\:x,P} \eqdef \!\!\!\!  \ET_{(x,x') \in E^*}\!\!\!\! \widehat{\vfi}^{\:x',P} \\
\widehat{\Ex \vfi \Until \psi}^{\:x,P} &\eqdef   \overline{\Ex \vfi \Until \psi}^{x,P,\{x\}} \quad \quad \text{with:} \\
 & \overline{\Ex \vfi \Until \psi}^{\:x,P,X} \eqdef \; \widehat{\psi}^{\:x,P} \ou \Big(\widehat{\vfi}^{\:x,P} \et \OU_{\stackrel{(x,x')\in E}{\scriptscriptstyle \text{s.t.}\: x'\not\in X}} \overline{\Ex \vfi \Until \psi}^{x',P,X\cup\{x'\}} \Big)  \\
\widehat{\All \vfi \Until \psi}^{\:x,P} &\eqdef   \overline{\All \vfi \Until \psi}^{\:x,P,\{x\}} \quad \quad \text{with:} \\
 & \overline{\All \vfi \Until \psi}^{\:x,P,X} \eqdef \; \begin{cases} \widehat{\psi}^{\:x,P}  \quad \quad \quad \text{if}\: \exists (x,x') \in E \text{ s.t. } x'\in X & \\ 
 {\displaystyle \widehat{\psi}^{\:x,P}  \ou \Big(\widehat{\vfi}^{\:x,P} \et \!\!\!\!\ET_{(x,x')\in E}  \!\!\!\!\! \overline{\All \vfi \Until \psi}^{x',P,X\cup\{x'\}} \Big)} & \!\!\!\!\text{otherwise} \\ \end{cases}
\end{align*}
\caption{Translation from \QCTL\ to \QBF}
\label{reduc-main}
\end{table}

Applied in a top down manner, the equivalences in Table~\ref{reduc-main} define a translation from  \QCTL\ to \QBF. Before stating the correctness of this  translation, we need to associate a Boolean valuation $v_\varepsilon$ for variables in $\AP_Q^\Phi\times V$ to an environment $\varepsilon$ for $\AP_Q^\Phi$.  We define $v_\varepsilon$ as follows: for any $p\in \AP_Q^\Phi$ and $x \in V$,  $v_\varepsilon(p^x) = \top$ iff $x \in \epsilon(p)$.   
Now we have the following theorem which establishes the correctness of the translation $\widehat{\phi}^{\:x,P}$ (its proof is in~\ref{app-uu}):
\begin{theorem}
\label{theo-UU}
Given a \QCTL\ formula $\Phi$, a Kripke structure    $\calK=\tuple{V,E,\ell}$, a state $x \in V$, an environment $\varepsilon : \AP_Q^\Phi \mapsto 2^V$ and a  $\Phi$-subformula $\vfi$, if  $\widehat{\vfi}^{x,\dom(\varepsilon)}$ is defined inductively w.r.t.\ the rules of Table~\ref{reduc-main}, we have:
$\calK, x \sat_\varepsilon \vfi \quad\mbox{iff}\quad v_\varepsilon \sat \widehat{\vfi}^{x,\dom(\varepsilon)}$
\end{theorem}

It is easy to deduce that we have: $\calK, x_0 \sat  \Phi$ iff $\widehat{\Phi}^{x_0,\emptyset}$ is valid. In the sequel we use $\widehat{\Phi}^{x_0}$ to denote $\widehat{\Phi}^{x_0,\emptyset}$.
The first reduction (called \MetUU) consists in applying the translation  directly on the \QCTL \ formula without any preprocessing: 

\begin{maboite}
\textbf{Reduction \MetUU:}
Given a Kripke structure $\calK$, a state $x$ in $\calK$ and a \QCTL\ formula $\Phi$, the reduction \MetUU\ (for $\calK,x \sat \Phi$) is defined as the \QBF\ formula $\widehat{\Phi}^{x}$. Its correction is a direct consequence of  Theorem~\ref{theo-UU}. The size of the \QBF\ formula is in $O((|\Phi|\cdot|V|!)^{\HT(\Phi)})$.
\end{maboite}

The main drawback of this naive reduction is the size of the \QBF\ formula (any Until modality may induce a formula whose size is in $O(|V|!)$). Nevertheless, one can notice that the reduction does not introduce new quantified propositions to encode the temporal modalities, contrary to other methods we will see later. 
Note also that this method can be adapted to  \emph{bounded} model-checking by fixing a bound on the number of unfoldings.

\subsection{Fixpoint characterization of the until operators}

Here we present the fixpoint method (called \MetFP) for dealing with the modalities $\All\Until$\ and $\Ex\Until$. Let $\phi$ and $\psi$ be two \QCTL-formulas. The idea of the method is to build a \QCTL\  formula that is equivalent to $\Ex\phi\Until\psi$ (or $\All\phi\Until\psi$) by using only the modalities  $\Ex\X$, $\All\X$ and  $\All\G$. We first have the following lemma:

\begin{lemma}
\label{lemFP1}
For any \QCTL\ formula $\Ex \vfi  \Until \psi$, we have:
\[
\Ex \vfi  \Until \psi \quad\equiv\quad \forall z. \Big( \AG \big( z \equivalent (\psi \ou (\vfi \et \EX \: z)) \big) \; \impl \; z\Big)
\]
\end{lemma}
\begin{proof}
Let $x$ be a state in a Kripke structure $\calK$. Let $\theta$ be the formula $\big(\AG \big( z \equivalent (\psi \ou (\vfi \et \EX \: z)) \big) \; \impl \; z\big)$. \\
Assume $\calK, x \sat \Ex \vfi  \Until \psi$. We can use the standard characterization of $\EU$ as fixpoint: $x$ belongs to the least fixpoint of the  equation $Z \eqdef \psi \ou (\vfi \et \EX \: Z)$ where $\psi$ (resp.\ $\vfi$) is here interpreted as the set of states satisfying $\psi$ (resp.\ $\vfi$). Therefore any $z$-labelling of reachable states from $x$~\footnote{Labelling other states does not matter.} corresponding to a fixpoint will have the state $x$ labelled. This is precisely what is specified by the \QCTL\ formula. \\
Now if $\calK, x \sat \theta$ for every $z$-labelling corresponding to a fixpoint of the previous equation, this is the case for the $z$-labelling of the states reachable from $x$ and satisfying $\Ex \vfi  \Until \psi$, and we deduce $\calK, x \sat \Ex \vfi  \Until \psi$.  
\hfill$\Box$
\end{proof}

And we have the same result for $\AU$ (whose proof is similar):
\begin{lemma}
\label{lemFP2}
For any \QCTL\ formula $\All \vfi  \Until \psi$, we have:
\[
\All \vfi  \Until \psi \quad\equiv\quad \forall z. \Big( \AG \big( z \equivalent (\psi \ou (\vfi \et \AX \: z)) \big) \; \impl \; z\Big)
\]
\end{lemma}

As a direct consequence, we get the following result:
\begin{proposition}
\label{prop-qctlres}
For any \QCTL\ formula  $\Phi$, we can build an  equivalent  \QCTL\ formula $\FPC(\Phi)$ such that:
(1) $\FPC(\Phi)$ is built up from atomic propositions, Boolean operators, propositional quantifiers and modalities $\EX$ and $\AG$, and 
(2) the size of   $\FPC(\Phi)$ is in $O(|\Phi|^{\HT(\Phi)})$.
\end{proposition}
The exponential size of $\FPC(\Phi)$ comes from the fact that  $\equivalent$ is not considered as a primitive of \QBF\ and then induces a duplication  of subformulas when applying the transformation rules based on equivalences of Lemmas~\ref{lemFP1} and~\ref{lemFP2} (otherwise its size would be linear in $|\Phi|$). Moreover, the temporal height of $\FPC(\Phi)$ is at most  $\HT(\Phi)+1$\footnote{The temporal height can be increased by 1 if $\Phi$ has an until operator whose subformulas are Boolean combinations of atomic propositions.}.

Now we can formally define the reduction \MetFP:

\begin{maboite}
\textbf{Reduction \MetFP:} Given a Kripke structure $\calK$, a state $x$ in $\calK$ and a \QCTL\ formula $\Phi$, the reduction \MetFP\ is defined as the \QBF\ formula $\widehat{\FPC(\Phi)}^{x}$. Its correction is a direct consequence of Proposition~\ref{prop-qctlres} and Theorem~\ref{theo-UU}. The size of the \QBF\ formula is $O((|\Phi|\cdot|\calK|)^{\HT(\Phi)+1})$.
\end{maboite}
The exponential size comes from the nesting of temporal modalities and each one may provide  a  \QBF\ formula of size $(|V|+|E|)$. Note also that the number of propositional variables in the \QBF\ formula is bounded by $|\Phi|\cdot|V|$.

\subsection{Reduction via flat formulas (\MetFPF)}

To avoid the size explosion of $\widehat{\Phi}^x$, one can use an alternative approach for prenex \QCTL\ formulas.
Remember that any \QCTL\ formula can be translated into an equivalent  \QCTL\ formula in prenex normal form whose size is linear in the size of the  original formula~\cite{LM14}. 

In the sequel, we use $S_\Phi$ to denote the set of \emph{temporal} subformulas occurring in $\Phi$ at a temporal depth greater than or equal to 1.

A \CTL\ formula is said to be \emph{basic} when it is of the form $\EX \alpha$, $\Ex \alpha \Until \beta$ or $\All \alpha \Until \beta$ where $\alpha$ and $\beta$ are Boolean combinations of atomic propositions (a basic formula is then formula starting with a temporal modality and whose temporal height is 1). 
It is easy to observe that any \CTL\ formula can be translated into a  \QCTL\ formula with a temporal height less or equal to 2:
\begin{proposition}
\label{prop-ctl-flat}
For any  \CTL\ formula $\Phi$, we can build an equivalent  \QCTL\ formula $\Psi$ of the form:
${\displaystyle 
\exists \kappa_1\ldots \exists \kappa_m. \Big( \Phi_0 \et \ET_{1\leq i \leq m} \AG (\kappa_i \Leftrightarrow   \theta_i)\Big)
}$
where $\Phi_0$ is a Boolean combination of basic \CTL\ formulas and every $\theta_i$ is a basic \CTL\  formula (for any $1\leq i \leq m$).  Moreover, $|\Psi|$ is in $O(|\Phi|)$. 
\end{proposition}
 
\begin{proof}
Consider a \CTL\ formula $\Phi$ built with temporal modalities in $\{\EX, \EU,$ $\AU\}$. 
The proof is done by induction over the size of $S_\Phi$. If $|S_\Phi|=0$, the original formula is a Boolean combination of basic \CTL\ formulas and it satisfies the property. Now, assume $|S_\Phi|>0$. $\Phi$ must have at least one basic (strict) subformula $\theta_1$. And $\Phi$ is equivalent to the formula $\exists \kappa_1.(\Phi[\theta_1\leftarrow \kappa_1] \et \AG (\kappa_1  \equivalent \theta_1))$, where $\kappa_1$ is a fresh atomic proposition, and $\vfi[\alpha\leftarrow \beta]$ is $\vfi$ where every occurrence of $\alpha$ is replaced by $\beta$. Indeed, any state reachable from the current state $x$ will be labelled by $\kappa_1$ iff $\theta_1$ holds true at that state (NB: the states that are not reachable from $x$ do not matter for the truth value of $\Phi$), and this enforces the equivalence. We have $|S_{\Phi[\theta_1\leftarrow \kappa_1]}|<|S_\Phi|$, thus we can apply induction hypothesis to get:
\[ \Phi[\theta_1\leftarrow \kappa_1] \equiv \exists \{\kappa_2\ldots \kappa_m\}. \Big( \Phi_0 \et \ET_{2\leq i \leq m} \AG (\kappa_i \Leftrightarrow   \theta_i)\Big)=\Psi' 
\]
where $\kappa_2\ldots\kappa_m$ are fresh atomic propositions, $\Phi_0$ is a Boolean combination of basic \CTL\ formulas and every $\theta_i$ is a basic \CTL\  formula. Then we have:
\begin{align*}
\Phi & \equiv \exists\kappa_{1}.\Big[  \Big(\exists \{\kappa_2\ldots \kappa_{m}\}. \Big( \Phi_0 \et \ET_{2 \leq i \leq  m} \AG (\kappa_i \Leftrightarrow   \theta_i) \Big) \Big) \et \AG (\kappa_1  \equivalent \theta_1)\Big] \\
  & \equiv  \exists \{\kappa_1\ldots \kappa_m\}. \Big( \Phi_0 \et \ET_{1\leq i \leq m} \AG (\kappa_i \Leftrightarrow   \theta_i) \Big)=\Psi
 \end{align*}
Note that the last    equivalence comes from the fact that no $\kappa_i$ with  $i>1$ occurs in  $\theta_1$.
By i.h. the size of $\Psi'$ is linear in $|\Phi[\theta_1\leftarrow \kappa_1]|$, and the size of $\Phi[\theta_1\leftarrow \kappa_1]$ is smaller than that of $\Phi$, therefore the size of $\Psi$ is linear in $|\Phi|$.
\hfill$\Box$
\end{proof}

And then we have:
\begin{proposition}
\label{prop-flat1}
For any  \QCTL\ formula $\Phi$, we can build an equivalent  \QCTL\ formula $\Flat_1(\Phi)$  of the form:
${\displaystyle \calQ. \big( \Phi_0 \et \ET_{1 \leq i \leq m} \AG (\kappa_i \Leftrightarrow   \theta_i)\big)}$
where $\calQ$ is a sequence of quantifications, $\Phi_0$ is a Boolean combination of basic \CTL\ formulas, every $\kappa_i$ is an atomic proposition, and every $\theta_i$ (for $i=1,\ldots,m$) is a basic \CTL\ formula
And $|\Flat_1(\Phi)|$ is in $O(|\Phi|)$. 
\end{proposition}
\begin{proof}
From~\cite{LM14}, we know how to build a prenex formula $\textsf{Prenex}(\Phi)=\calQ.\Gamma$ whose size is linear in $|\Phi|$. $\Gamma$ belongs to \CTL, and then Proposition~\ref{prop-ctl-flat} allows us to build a \QCTL\ formula ${\displaystyle \calQ . \exists \kappa_1\ldots \exists \kappa_m. \Big( \Phi_0 \et \ET_{1\leq i \leq m} \AG (\kappa_i \Leftrightarrow   \theta_i)\Big)}$ equivalent to $\Phi$. Note that we have $\HT(\Flat_1(\Phi))\leq 2$.
\end{proof}

We can define a new reduction:

\begin{maboite}
\textbf{Reduction \MetFPF:} Given a Kripke structure $\calK$, a state $x$ in $\calK$ and a \QCTL\ formula $\Phi$, the reduction \MetFPF\ is defined as the \QBF\ formula $\widehat{\FPC(\Flat_1(\Phi))}^{x}$. Its correction is a direct consequence of Proposition~\ref{prop-flat1}, Proposition~\ref{prop-qctlres} and Theorem~\ref{theo-UU}. The size of the \QBF\ formula is $O((|\Phi|\cdot|\calK|)^{3})$.
\end{maboite}

Therefore this reduction   provides a \PSPACE\ algorithm for   \QCTL\ model-checking. 
But there are two disadvantages to this approach. First, putting the formula into prenex normal form may increase the number of quantified atomic propositions and the number of alternations (which is \emph{in fine} linear in the number of quantifiers in the original formula) \cite{LM14}. For example, when extracting a quantifier $\forall$ from some $\EX$ modality, we need to introduce two propositions, this can be seen for  the formula $\EX (\forall p. (\AX p \ou \AX \non p))$ which  is translated as:
\[
 \exists z. \forall p. \forall z'. \Big(    (\EX (z \et z') \impl \AX(z \impl z')) \et  \EX (z \et   (\AX p \ou \AX \non p))\Big)
\]
where the proposition $z$ is used to mark a state, and $z'$ is used to  enforce that only at most one successor is labelled by $z$. 
Of course, these two remarks may have a strong impact on the complexity of the decision procedure. 
Finally, note also that the resulting \QBF\ formula is not in prenex normal form because $\FPC$ may insert new quantifications. 

\begin{example}
To illustrate the reduction, we consider the following  (\CTL) formula $\Phi=\EF ( \Ex a \Until b \et \EX (\EX c))$.
The formula $\FPC(\Flat_1(\Phi))$  will provide the following \QCTL\ formula:
\begin{align*}
\exists k_1, k_2, k_3 . & \Big(  [ \forall  k_4. \AG (k_4 \equivaut ((k_1 \et k_3) \ou \EX\: k_4)) \impl k_4 ] \; \et \\
 & \AG \Big(k_1 \equivaut  [ \forall k_5. \AG (k_5 \equivaut (b \ou (a \et \EX \: k_5))) \impl k_5] \Big) \; \et \\
  & \AG (k_2 \equivaut \EX \: c) \; \et \; \AG \Big(k_3 \equivaut \EX \: k_2 \Big)  \Big)
 \end{align*}
where $k_1$ is used to label states satisfying $\Ex a \Until b$, $k_2$ is for $\EX c$, $k_3$ for $\EX \EX c$ and $k_4$ is used for the fixpoint encoding of the main $\EF$ modality. 
\end{example}

\subsection{Method \MetPNF}

We now propose a reduction to get  a \QBF\ formula in prenex normal form. It is also based on a flattening of the formula (but slightly different from the  one used for \MetFPF), and uses a different technique to eliminate temporal modalities $\EU$\ and $\AU$.

First, we consider a $\QCTL$\ formula $\Phi$ under  negation normal form (NNF): This transformation causes $\Phi$ to be built from temporal modalities in $S_{tmod} = \{\EX, \AX, \EU,$ $\AU, \EW, \AW\}$.
Due to the NNF, we can consider a new flattening process where equivalences are replaced by implications: we just need to ensure that if a $\kappa_i$ is true in some state $x$ (and possibly makes  some subformula  true), then the formula $\theta_i$ associated with $\kappa_i$ in the flattening actually holds true at $x$.
We have  the following proposition, whose proof is in~\ref{app-meth5}: 
\begin{proposition}
\label{prop-flat2}
For any  \QCTL\ formula $\Phi$, we can build an equivalent  \QCTL\ formula $\Flat_2(\Phi)$ in NNF and of the form:
\[{\displaystyle 
 \calQ \:  \exists \kappa_1\ldots \exists \kappa_m. \Big( \Phi_0 \et \ET_{1 \leq i \leq  m} \AG (\kappa_i \impl   \theta_i)\Big)
}\]
where $\calQ$ is a sequence of quantifications, $\Phi_0$ is a \CTL\ formula containing only the temporal modalities $\EX$, $\AX$, $\EF$ or $\AG$, and whose temporal height  is less than or equal to 1, and  every  $\theta_i$ is a basic \CTL\ formula (with $1 \leq i \leq m$).  Moreover, $|\Flat_2(\Phi)|$ is in $O(|\Phi|)$ and $\HT(\Flat_2(\Phi))\leq 2$.
\end{proposition}

From the previous proposition, we  derive a new reduction based on a transformation (denoted $\ReplaceUW(-)$) of $\Flat_2(\Phi)$: first we will add a new quantified proposition ($\chi$) to handle least fixpoints in $\Flat_2(\Phi)$ and then we will rewrite every 
$\Flat_2(\Phi)$-subformulas of the form $\AG (\kappa_i \impl   \theta_i)$ depending on the type of $\theta_i$ in order to get a formula built from temporal modalities in $\{\EX, \AX, \EF, \AG\}$. 

Assume ${\displaystyle \Psi = \calQ \:  \exists \kappa_1\ldots \exists \kappa_m. \Big( \Phi_0 \et \ET_{1\leq i \leq m} \AG (\kappa_i \impl   \theta_i)\Big)}$ with the same form as in Proposition~\ref{prop-flat2}, we define 
$\ReplaceUW(\Psi)$ as: 
\[
\calQ \:  \exists \{\kappa_1\ldots \kappa_m\}. \forall \chi.  \Big( \Phi_0 \et \ET_{1\leq i \leq m} \wideparen{\AG (\kappa_i \impl   \theta_i)} \Big)
\]
where the transformation $\wideparen{\theta_i}$ is defined by:
\begin{align*}
\wideparen{\AG (\kappa_i \impl \EX \psi)} & \eqdef \;  \AG (\kappa_i \impl \EX \psi) \\
\wideparen{\AG (\kappa_i \impl \AX \psi)} & \eqdef \;  \AG (\kappa_i \impl \AX \psi) \\
\wideparen{\AG (\kappa_i \impl \Ex \vfi \WUntil \psi)} & \eqdef \; \AG(\kappa_i \impl (\psi \ou (\vfi \et \EX \:\kappa_i))) \\
\wideparen{\AG (\kappa_i \impl \All \vfi \WUntil \psi)} & \eqdef \; \AG(\kappa_i \impl (\psi \ou (\vfi \et \AX \:\kappa_i))) \\
\wideparen{\AG (\kappa_i \impl \Ex \vfi \Until \psi)} & \eqdef \;   \EF((\psi \ou (\vfi \et \EX \:\chi)) \et \non \chi) \ou \AG (\kappa_i \impl \chi) \\
\wideparen{\AG (\kappa_i \impl \All \vfi \Until \psi)} & \eqdef \;  \EF((\psi \ou (\vfi \et \AX \:\chi)) \et \non \chi) \ou \AG (\kappa_i \impl \chi) \\
 \end{align*}

The correctness of this transformation is stated as follows:
\begin{proposition}
\label{prop-pnf}
For any \QCTL\ formula $\Phi$, we have: $\Phi \equiv \ReplaceUW(\Flat_2(\Phi))$.
\end{proposition}
\begin{proof}
We have to   prove that the substitutions are correct. 
The proof is easy for  weak Until modalities: for example, consider a state labelled by  $\kappa_i$ associated with some $\Ex \vfi \WUntil \psi$, then the substitution ensures that $x$ satisfies either (1) $\psi$ or (2) $\vfi$ with a successor satisfying $\kappa_i$ which ensures\ldots Clearly such  states labelled by $\kappa_i$ belong to the corresponding greatest fixpoint. 

Now consider the case of the transformation of $\AG (\kappa_i \impl \Ex \vfi \Until \psi)$. The formula $\EF((\psi \ou (\vfi \et \EX \:\chi)) \et \non \chi) \ou \AG (\kappa_i \impl \chi)$ specifies that either the $\chi$-labeling does not mark all states satisfying $\Ex \vfi \Until \psi$ (because there exists a reachable state which is not labelled by $\chi$ when it should because it satisfies $\psi \ou (\vfi \et \EX \:\chi)$) or every reachable state labelled by $\kappa_i$ is also labelled by $\chi$. Now if we consider that this property should be true for \emph{every} $\chi$-labelling (cf.\ the $\forall$ quantifier in the definition of $\ReplaceUW(\Flat_2(\Phi))$), then we can deduce that every state labelled by $\kappa_i$ has to belong to the least fixpoint (cf.\ Prop.\ 1, page 97 in~\cite{stirling2001}) and this ensures that $\Ex \vfi \Until \psi$ holds true for every state labelled by $\kappa_i$. 
\hfill$\Box$
\end{proof}

\begin{example}
If we consider the  formula $\Phi=\EF ( \Ex a \Until b \et \All c \WUntil (\EX d))$.
We have: $\Flat_2(\Phi) \eqdef \exists k_1 \:  k_2 \: k_3 .  \Big( \EF (k_1 \et k_3) \et \AG(k_1 \impl \Ex a \Until b) \et  \AG(k_2 \impl \EX d) \et  \AG(k_3 \impl \All c \WUntil k_2)\Big)$. 
And then we have:  
\begin{align*}
\ReplaceUW(\Flat_2(\Phi)) \eqdef & \exists k_1 \:  k_2 \: k_3 . \forall \chi .  \Big( \EF (k_1 \et k_3) \et  \\
 & \big[\EF (((b \ou (a \et \EX (k_4))) \et \non k_4)) \ou \AG (k_1 \impl k_4) \big] \et \\
 & \AG (k_2 \impl \EX d) \et \AG \big[ k_3 \impl \big( k_2 \ou (c \et \AX( k_3)) \big) \big]   \Big)
 \end{align*}
 \end{example}

And we get a new reduction providing a smaller formula. Indeed translating  formulas of the form $\AG(p \impl \EX q)$ in \QBF  \ provide a formula whose size is in $O(|V|+|E|)$, that is in $O(\calK)$.

\begin{maboite}
\textbf{Reduction \MetPNF:} Given a Kripke structure $\calK$, a state $x$ in $\calK$ and a \QCTL\ formula $\Phi$, the reduction \MetPNF\ is defined as the prenex \QBF\ formula 
\stackon[-8pt]{$ \ReplaceUW(\Flat_2(\Phi))$}{\vstretch{1.5}{\hstretch{16}{\widehat{\phantom{\;}}}}}$^{~^{\scriptstyle \!x}}$.
Its correction is a direct consequence of Proposition~\ref{prop-pnf} and Theorem~\ref{theo-UU}. The size of the \QBF\ formula is $O(|\calK|\cdot |\Phi|)$.
\end{maboite}

Again this provides another algorithm in \PSPACE.

\subsection{Method based on BitVec to encode distance (\MetFBV)}

In the previous reduction, the modalities $\EU$ and $\AU$ may introduce an alternation of quantifiers: 
an  atomic proposition $\kappa$ is introduced by an existential quantifier, and then a universal quantifier introduces a variable $\chi$ to encode the fixpoint characterisation of $\Until$. We propose another reduction in order to avoid this alternation: for this, we will use bit vectors (instead of single Boolean values)  to encode the \emph{distance} (in terms of number of transitions)
 from the current state to a state satisfying the right-hand side of the Until modality. 
 To build such a formula, we have  to fix the size of the bit vectors and then the maximum value we can represent, therefore   
 the reduction is parameterised by a value $N$   and the bit vectors will be able to encode values from $0$ to $N$ ($N$ will be set by $|V|$ when the reduction is applied over a Kripke structure). 
 Therefore contrary to previous reductions, the correctness of the preprocessing will depend on the size (the number of states) of the structure we will consider. 
  
 As for \MetPNF, we consider 
  a $\QCTL$\ formula $\Phi$ under  negation normal form (NNF) and we  use the $\Flat_2$ transformation, but we  define a new transformation, called  $\ReplaceUW^2(\Psi,N)$, to replace temporal modalities. 

For modalities $\EW$ and $\AW$, we use the same idea as in $\ReplaceUW$ for method \MetPNF. For the Until-based modalities, corresponding to least fixpoints, we  use \emph{bit vectors} of length $\kbv$~\footnote{The value of $\kbv$ will depend on $N$.} instead of a single Boolean value  $\kappa_i$ 
 to encode the truth value of $\EU$ or $\AU$. For  a formula  $\theta_i=\Ex \vfi \Until \psi$, 
 the value $\bar{\kappa_i}$ encodes in binary the  distance  from $x$ to a state satisfying $\psi$ along some path satisfying  $\vfi \Until \psi$, and  for $\theta_i=\All \vfi \Until \psi$, the value $\bar{\kappa_i}$ encodes an overapproximation of the distance before reaching a state satisfying $\psi$ (and verifying $\vfi \Until \psi$) along any path issued from $x$.

In the new reduction we need to compare the values encoded by  bit vectors  with  integer values encoded in binary. These comparisons will be done by propositional formulas. Consider 
 a bit vector $\bar{\kappa_i} = \kappa_i^{\kbv-1} \cdots \kappa_i^0$ and an integer value $d$ encoded in binary with $\kbv$ bits $d^{\kbv-1}\cdots d^0$ where $\kappa_i^0$ and $d^0$ are the least significant digits of the values. In the following we identify 1 with $\top$ and 0 with $\bot$. We define two types of formulas $[\bar{\kappa_i}=d]$ (\ie\ the value encoded by $\bar{\kappa_i}$ equals to $d$) and $[\bar{\kappa_i}<d]$ (\ie\ the value encoded by $\bar{\kappa_i}$ is  less than $d$) as follows:
 \[
 [\bar{\kappa_i}=d] \eqdef \!\!\! \ET_{0\leq j < \kbv} \!\!\! \!\!  \kappa_i^j \equivaut d^j \quad \quad  
 [\bar{\kappa_i}<d] \eqdef  \!\!\!\OU_{0\leq m < \kbv}\!\!\! (\non \kappa_i^m \et d^m \et \!\!\!\ET_{m< j < \kbv} \!\!\!(\kappa_i^j \equivaut d^j) 
 \]
 
Assume ${\displaystyle \Psi = \calQ \:  \exists \kappa_1\ldots \exists \kappa_m. \Big( \Phi_0 \et \ET_{1\leq i \leq m} \AG (\kappa_i \impl   \theta_i)\Big)}$ with the same form as in Proposition~\ref{prop-flat2}, we define 
$\ReplaceUW^2(\Psi,N)$ as: 
\[
\calQ \:  \exists \{\dot{\kappa_1} \ldots \dot{\kappa_m}\} .  \Big( \widetilde{\Phi_0} \et \ET_{1\leq i \leq m} \widetriangle{\AG (\kappa_i \impl   \theta_i)} \Big)
\]
where:
\begin{itemize}
\item $\kbv =\lceil\log(N+1)\rceil$, 
\item  $\dot{\kappa_i}$ is a Boolean proposition (resp.\ a vector of $\kbv$ Boolean propositions) if $\theta_i$ is of the form $\EX \psi$, $\AX \psi$, $\Ex \vfi \WUntil \psi$ or $\All \vfi \WUntil \psi$ (resp.\ if $\theta_i$ is of the form $\Ex \vfi \Until \psi$ or $\All \vfi \Until \psi$), 
\item 
the transformation $\widetriangle{\theta_i}$ is defined as follows:
\begin{align*}
\widetriangle{\AG (\kappa_i \impl \EX \psi)} & \eqdef \; \AG(\kappa_i \impl \EX \:\widetilde{\psi}) \\
\widetriangle{\AG (\kappa_i \impl \AX \psi)} & \eqdef \; \AG(\kappa_i \impl \AX \:\widetilde{\psi}) \\
\widetriangle{\AG (\kappa_i \impl \Ex \vfi \WUntil \psi)} & \eqdef \; \AG(\kappa_i \impl (\widetilde{\psi} \ou (\widetilde{\vfi} \et \EX \:\kappa_i))) \\
\widetriangle{\AG (\kappa_i \impl \All \vfi \WUntil \psi)} & \eqdef \; \AG(\kappa_i \impl (\widetilde{\psi} \ou (\widetilde{\vfi} \et \AX \:\kappa_i))) \\
\widetriangle{\AG (\kappa_i \impl \Ex \vfi \Until \psi)} & \eqdef \;   \AG \big[ ( [\bar{\kappa_i}=0] \impl \:\widetilde{\psi}) \et \\ & \hfill \ET_{1 \leq d < N} ([\bar{\kappa_i}=d] \impl \: (\widetilde{\vfi} \et \EX\: [\bar{\kappa_i}=d-1])) \big] \\
\widetriangle{\AG (\kappa_i \impl \All \vfi \Until \psi)} & \eqdef \;  \AG \big[ ( [\bar{\kappa_i}=0] \impl \:\widetilde{\psi}) \et \\ & \hfill \ET_{1 \leq d < N} ([\bar{\kappa_i}=d] \impl \: (\widetilde{\vfi} \et \AX\: [\bar{\kappa_i}<d])) \big] 
 \end{align*}
 \item $\widetilde{\alpha}$ denotes the formula $\alpha$ where every occurrence of $\kappa_i$ is replaced by the formula $[\bar{\kappa_i}^x < N]$ for every $i$ such that  $\dot{\kappa_i}$ is a bit vector (\ie\ is associated with $\EU$ or $\AU$ in the flattening step), 
 \end{itemize}

And we have the following proposition which specifies that $\ReplaceUW^2(\Flat_2($ $\Phi),N)$ is equivalent to $\Phi$ for every structure whose size is at most $N$:
\begin{proposition}
\label{prop-fbv}
For any \QCTL\ formula $\Phi$, any Kripke structure $\calK=(V,E,\ell)$ with $|V|\leq N$ and any $x \in V$, we have: $\calK, x \sat \Phi$ if and only if  $\calK, x \sat \ReplaceUW^2(\Flat_2(\Phi),N)$.
\end{proposition}
\begin{proof}
Consider the formula $\Ex \vfi \Until \psi$ and some bit vector $\bar{\kappa_i}$ associated with it.
If $\calK, x \sat \Ex \vfi \Until \psi$, then there exists a path $\rho$ from $x$ satisfying $\vfi \Until \psi$. The length of the prefix of $\rho$ from $x$ to the state $y$ where $\psi$ is true is at most $|V|-1$ (we can assume the path to be simple). Therefore we can fix the values of $\bar{\kappa_i}$ for all states along the prefix to the distance to $y$ and we have the result. And for the states that do not satisfy $\vfi \Until \psi$, we can set $\bar{\kappa_i}$ to  the value $N$. \\
For the other direction, the definition of a $\bar{\kappa_i}$  ensures that if the value $v$ encoded by $\bar{\kappa_i}$ in a given state $x$  is less than $N$, then $v$ corresponds actually to the  distance from $x$ to  a state satisfying $\psi$ along a path satisfying $\vfi \Until \psi$ (it decreases down to 0): this distance is maybe not the length of a shortest path but it is sufficient to deduce that $\Ex \vfi \Until \psi$ holds true at $x$. 
 And indeed when the truth value of $\Ex \vfi \Until \psi$ in a state $x$ is needed, we interpret  $[\bar{\kappa_i}^x < N]$ at $x$. 

For $\All \vfi \Until \psi$, we can see from the definition that if the value $v$ encoded by $\bar{\kappa_i}$ in  $x$  is less than $N$, then all  successors of $x$ have an encoded value less than  $v$, which ensures that $\All \vfi \Until \psi$ holds true at $x$. 

Note that when the value encoded by $\bar{\kappa_i}$ is greater than $N$, then nothing is ensured and in both case, the subformula will not be assumed to be true. 
\hfill$\Box$
\end{proof}

Finally we can see that the transformation $\ReplaceUW^2(-,N)$ increases the size of the formula by a factor of $N\cdot \lceil\log(N+1)\rceil$. 
And then we have:

\begin{maboite}
\textbf{Reduction \MetFBV:} Given a Kripke structure $\calK$, a state $x$ in $\calK$ and a \QCTL\ formula $\Phi$, the reduction \MetFBV\ is defined as the prenex \QBF\ formula 
\stackon[-8pt]{$ \ReplaceUW^2(\Flat_2(\Phi),|V|)$}{\vstretch{1.5}{\hstretch{16}{\widetriangle{\phantom{\;}}}}}$^{~^{\scriptstyle \!x}}$.
Its correction is a direct consequence of Proposition~\ref{prop-fbv} and Theorem~\ref{theo-UU}. The size of the \QBF\ formula is in $O(|V| \cdot |\calK| \cdot \lceil\log(|V|+1)\rceil  \cdot |\Phi|)$.
\end{maboite}

This provides another \PSPACE\ algorithm but  the size of the resulting \QBF~ formula  is larger than that obtained by the  previous method. Finally one can notice that this approach would allow us  to easily adapt the algorithm to \emph{bounded} model-checking: instead of considering values from $1$ to $|V|$ for $d$ in the definition of Until modalities, one could restrict the range to a smaller interval, to get a smaller \QBF\ formula to check. 

\subsection{Dealing with $\exists^1$ and  $\forall^1$.}

\label{sec-ea1}

The quantifiers $\exists^1$ and $\forall^1$ are very useful in many specifications. It can be interesting to develop ad-hoc algorithms  in order to improve the generated \QBF\ formulas and to be able to choose the encoding of these operators. 
Formally we defined $\exists^1p.\vfi$ as $\exists p.((\Ex_{=1}\F p) \et \vfi)$. In the following, we use the abbreviation $\uniq(p)$ to denote $\Ex_{=1}\F p$ and we will see several methods to deal with it during the translation into \QBF.

There are three possible encodings of $\uniq(p)$:
\begin{enumerate}
\item by using its \QCTL\ definition ($\Ex_{=1}\F p$): and translate this formula with the rules described above for the different methods.
\item by an explicit disjunction in the \QBF \ formula: given a state $x$, $\uniq(p)$ is equivalent to:
${\displaystyle  \Big( \OU_{(x,y) \in E^*} (p^y \et \ET_{z \not= y} \non p^z)\Big)}$.
\item by Bit vectors: the quantifier ($\exists$ or $\forall$) associated with  $p$  introduces  a  \emph{bit vector} $\bar{p}$ of size $\lceil\log(|V|+1)\rceil$ to store the \emph{number} of the state selected by the quantifier: in the \QBF\ formula, verifying that a state $x$ is labelled by $p$ consists in verifying that the number of $x$ equals to the value $\bar{p}$ (which can be encoded as a propositional formula). And $\uniq(p)$ consists in verifying that there exists one reachable state whose number is $\bar{p}$. Note that  this method reduces the number of quantified propositions (instead of having $|V|$ propositions as in the two first methods, we have only $\lceil\log(|V|+1)\rceil$ propositions with this encoding). 
\end{enumerate}

Note that the encoding may provide a formula that is not anymore in PNF. To avoid this, we can gather all subformulas $\uniq(-)$ in the main part of the \QCTL\ formula in order to keep the prenex form. Indeed we can first observe the following equivalences when $\vfi$ does not depend on $p$:
\begin{align*}
\vfi \:\et \exists p.\psi  & \;\equiv \;  \exists p. (\vfi \:\et\: \psi)   \quad \quad \vfi \:\et\: \forall p.\psi  \;\equiv \;   \forall p. (\vfi \:\et\: \psi)   \\
\vfi \:\impl\: \exists p.\psi  & \;\equiv \;   \exists p. (\vfi \:\impl\: \psi)   \quad \quad \vfi \:\impl\: \forall p.\psi  \;\equiv \;   \forall p. (\vfi \:\impl\: \psi)   \\
\end{align*}

Now consider a Prenex \QCTL\ formula $\calQ.\Phi$ where $\calQ$ contains  $\exists^1$ or $\forall^1$ quantifications for $k$ propositions $p_1,\ldots,p_k$. Let $P_{\exists^1} \subseteq \{1,\ldots,k\}$ be the set of indexes such that  $i\in P_{\exists^1}$ iff $p_i$ is introduced by some quantifier $\exists^1$ in $\Phi$. Let $P_{\forall^1}$ be $\{1,\ldots,k\} \setminus P_{\exists^1}$.
The previous equivalences allow us to get a formula of the form:
\[
\Phi' \eqdef \widetilde{\calQ} . (\uniq(p_1) \:op_1\: (\uniq(p_2) \:op_2\: (\; \ldots \;  \uniq(p_k) \:op_k \: \Phi)))
\]
where (1) $\widetilde{\calQ}$ is $\calQ$ where $\exists^1$ (resp.\ $\forall^1$) is replaced by  $\exists$ (resp.\  $\forall$), and (2) $op_i$ is $\et$ (resp.\ $\impl$) if $i\in P_{\exists^1}$ (resp.\  $i \in P_{\forall^1}$).
And it remains to see that the formula $\Phi'$  is equivalent to:
\[
\Phi'' \eqdef \calQ . \Big( (\ET_{i \in P_{\forall^1}} \uniq(p_i)) \:\impl\: \Big( (\ET_{i\in P_{\exists^1}} \uniq(p_i)) \et \Phi\Big)\Big)
\]
The correctness of this equivalence is based on  Proposition~\ref{prop-ea1} in~\ref{app-ea1}.

Finally, the three previous encodings  can be used to get a \QBF\ formula. Note that the last two (explicit disjunction and bit vectors) will then provide a formula in prenex normal form. Moreover the first and third encodings use the structure $\calK$.

Note also that we could also consider a variant with $\exists^{\leq 1}$ where the associated proposition  cannot label more than one state. The same encodings described above can be done.


\section{Experimental results}

\label{sec-experiments}

In this section, we consider four examples to evaluate and compare the different reductions.   
 These problems  can be solved  efficiently without a \QCTL\ model-checker, but they provide valuable insights on the performance of the reduction strategies, and the properties to be checked cannot be expressed  with classical temporal logics. In every case we consider graphs that can easily be scaled up by tweaking a few parameters.

For these experiments, we use our prototype \qctlmc~\footnote{Our tool is available online \url{https://www.irif.fr/~francoisl/qctlmc.html}, it implements the different reductions.}
 to translate model-checking instances for \QCTL\ into \QBF\ instances. The reductions produce either a \QBF\ formula in QCIR-G14 format~\footnote{See \url{http://www.qbflib.org}} or in the format used by  the tool \Z3. We considered the following \QBF\ solvers~\footnote{Many   \QBF\ solvers handle only prenex CNF formulas, but here we need non CNF formula and we didn't want to add  extra formula rewritings, this explains our choice of solvers.}:
\begin{itemize}
\item \Z3: it is a powerful  SMT solver~\cite{MouraB08} which also handles \QBF \ instances. There are many features in \Z3 but we only use the restricted part for \QBF. Its admits any kind of \QBF\ formula in an adhoc format (called \Z3 format here). We used \Z3 \ 4.8.7 for the tests.
\item \qfm: it is a \QBF \ solver~\cite{qfm} with counterexample guided refinement (based on sat-solver cadical~\cite{Biere-SAT-Race-2019-solvers} or minisat\cite{EenS03}), it allows us to deal with general \QBF\ formulas in QCIR format. 
\item \qfun~\footnote{\url{http://sat.inesc-id.pt/~mikolas/sw/qfun/}}:  it is a \QBF \ solver based on Recursive Abstraction Refinement  and machine learning  \cite{qfun}. It requires prenex formulas in QCIR format.
\item \cqesto~\footnote{\url{http://sat.inesc-id.pt/~mikolas/sw/cqesto/}}: it is a QBF solver based on clause selection~\cite{cqesto}, it requires a prenex formula in QCIR format. 
\end{itemize}

The solvers \cqesto\ and \qfun \ require prenex formulas and then can only be used with the reductions \MetPNF\  and \MetFBV. There are none of these restrictions for   \Z3\ and \qfm. All results are presented in Subsection~\ref{sec-results}.

\subsection{Reset property}

\label{sec-scc}

We  consider the \emph{reset property}: the existence of a set of (at most) $m$ states such that from any reachable state it is possible to reach at least one of the selected states. For this we can use $\Lambda_{m} \eqdef {\displaystyle \exists^1 p_1 \ldots \exists^1 p_m . \big(\AG (\EF \OU_{1\leq i \leq m} p_i)\big)}$ which  selects $m$ states by using the propositions $p_i$s (the use of quantifier $\exists^1$ ensures that at most $m$ states are selected).  

Now given two parameters $n,k \in \Nat$, we define the Kripke structure  $\calV_{n,k}=(V,E,\ell)$ that contains a root $r$ and $n$ different cycles of length $k$ (and no atomic proposition). It is depicted at Figure~\ref{fig-scc}.
We then clearly have $\calV_{n,k} \sat \Lambda_{m}$ iff $m\geq n$. The main characteristics of this example are: a simple temporal formula, the use of the $\exists^1$ operator and a Kripke structure with a low branching degree.

\begin{figure}[t]
\centering
\begin{tikzpicture}[inner sep=0pt,scale=0.8]
\draw (-1.8,2.1) node[moyrond,vert] (r) {$r$} ;
\draw (0,4) node[minirond,vert] (v00) {\footnotesize $q_{1,1}$} ;
\draw (1.7,4) node (v01) {\ldots} ;
\draw (3.4,4) node[minirond,vert] (v02) {\footnotesize $q_{1,k}$} ;
\draw (0,2.5) node[minirond,vert] (v10) {\footnotesize $q_{2,1}$} ;
\draw (1.7,2.5) node (v11) {\ldots} ;
\draw (3.4,2.5) node[minirond,vert]   (v12) {\footnotesize $q_{2,k}$}  ;
\draw (0,0) node[minirond,vert] (v30) {\footnotesize $q_{n,1}$} ;
\draw (1.7,0) node (v31) {\ldots} ;
\draw (3.4,0) node[minirond,vert] (v32) {\footnotesize $q_{n,k}$} ;
\draw (0,1.3) node (v20) {} ;
\draw (1.5,1.3) node (v21) {\ldots} ;

\draw[-latex'](r) -- (v00) ;
\draw[-latex'](r) -- (v10) ;
\draw[-latex'](r) -- (v20) ;
\draw[-latex'](r) -- (v30) ;

\draw[-latex'](v00) -- (v01) ;
\draw[-latex'](v01) -- (v02) ;
\draw[-latex'](v10) -- (v11) ;
\draw[-latex'](v11) -- (v12) ;
\draw[-latex'](v30) -- (v31) ;
\draw[-latex'](v31) -- (v32) ;
\draw[-latex',out=20,in=160] (v02) to  (v00);
\draw[-latex',out=20,in=160] (v12) to  (v10);
\draw[-latex',out=-20,in=-160] (v32) to  (v30);

\draw (2,-1) node  (titre) { $\calV_{n,k}$} ;

\begin{scope}[shift={(5,0)},inner sep=0pt]

\draw (0,4) node[minirond,vert] (v00) {\footnotesize $q_{1,1}$} ;
\draw (2,4) node[minirond,vert] (v01) {\footnotesize $q_{1,2}$} ;
\draw (4,4) node (v02) { $\;\ldots\;$} ;

\draw (6,4) node[minirond,vert] (v04) {\footnotesize $q_{1,k}$} ;
\draw (0,2.5) node[minirond,vert] (v10) {\footnotesize $q_{2,1}$} ;
\draw (2,2.5) node[minirond,vert] (v11) {\footnotesize $q_{2,2}$} ;
\draw (4,2.5) node  (v12) { $\;\ldots\;$} ;
\draw (6,2.5) node[minirond,vert] (v14) {\footnotesize $q_{2,k}$} ;
\draw (0,0) node[minirond,vert] (v30) {\footnotesize $q_{n,1}$} ;
\draw (2,0) node[minirond,vert] (v31) {\footnotesize $q_{n,2}$} ;
\draw (4,0) node (v32) { $\;\ldots\;$} ;
\draw (6,0) node[minirond,vert] (v34) {\footnotesize $q_{n,k}$} ;
\draw (0,1.3) node  (v20) { $\;\ldots\;$} ;
\draw (2,1.3) node  (v21) { $\;\ldots\;$} ;
\draw (4,1.3) node  (v22) { $\;\ldots\;$} ;
\draw (6,1.3) node  (v24) { $\;\ldots\;$} ;

\draw (7,4.7) node (aux11) { } ;
\draw (7,4) node (aux12) { } ;
\draw (7.5,4) node (aux02b) {$q_{i,1}$ } ;
\draw (7,3.3) node (aux13) { } ;
\draw[-latex',dashed] (v04) to  (aux11);
\draw[-latex',dashed] (v04) to  (aux12);
\draw[-latex',dashed] (v04) to  (aux13);

\draw (7,3.2) node (aux01) { } ;
\draw (7,2.5) node (aux02) {} ;
\draw (7.5,2.5) node (aux02b) {$q_{i,1}$ } ;
\draw (7,1.8) node (aux03) { } ;
\draw[-latex',dashed] (v14) to  (aux01);
\draw[-latex',dashed] (v14) to  (aux02);
\draw[-latex',dashed] (v14) to  (aux03);

\draw (7,0.7) node (aux11) { } ;
\draw (7,0) node (aux12) { } ;
\draw (7.5,0) node (aux02b) {$q_{i,1}$ } ;
\draw (7,-0.7) node (aux13) { } ;
\draw[-latex',dashed] (v34) to  (aux11);
\draw[-latex',dashed] (v34) to  (aux12);
\draw[-latex',dashed] (v34) to  (aux13);

\draw[-latex](v00) -- (v01) ;
\draw[-latex](v00) -- (v11) ;
\draw[-latex](v00) -- (v31) ;

\draw[-latex](v01) -- (v02) ;
\draw[-latex](v01) -- (v12) ;
\draw[-latex](v01) -- (v22) ;
\draw[-latex](v01) -- (v32) ;

\draw[-latex](v10) -- (v11) ;
\draw[-latex](v10) -- (v01) ;
\draw[-latex](v10) -- (v31) ;

\draw[-latex](v30) -- (v11) ;
\draw[-latex](v30) -- (v01) ;
\draw[-latex](v30) -- (v31) ;

\draw[-latex](v11) -- (v12) ;
\draw[-latex](v12) -- (v14) ;
\draw[-latex](v31) -- (v32) ;
\draw[-latex](v32) -- (v34) ;

\draw[-latex](v02) -- (v04) ;

\draw (4,-1) node  (titre) { $\calK_{n,k}$} ;
\end{scope}

\end{tikzpicture}
\vspace{-0.2cm}
\caption{Structures for  the reset property and the resources distribution.}
\label{fig-scc}\label{fig-resources}
\end{figure}
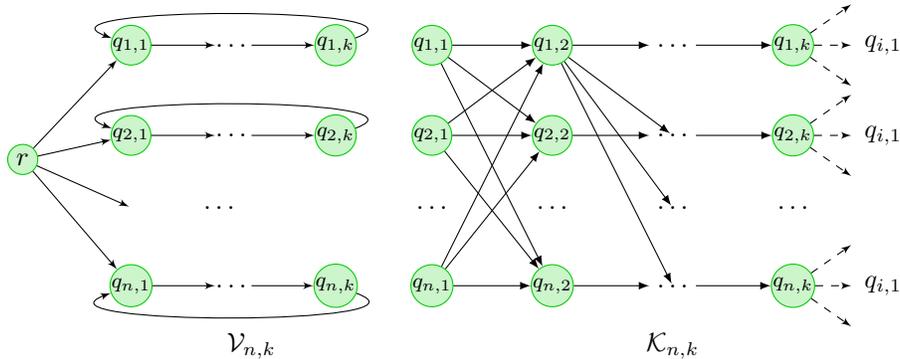

\subsection{$k$-connectivity}

Here, we consider an undirected graph, and we want to check whether there exist (at least) $k$ internally disjoint paths~\footnote{Two paths $src \leftrightarrow r_1 \leftrightarrow \ldots \leftrightarrow r_k \leftrightarrow dest$  and $src \leftrightarrow r'_1 \leftrightarrow \ldots \leftrightarrow r'_{k'} \leftrightarrow dest$  are internally disjoint iff $r_i  \not= r'_{j}$ for any $1 \leq i \leq k$ and $1 \leq j \leq k'$. And note that with this definition, if there is an edge $(x,y)$, there exist $k$ internally disjoint paths from $x$ to $y$ for any $k$.} from a vertex $x$ to some vertex $y$.
A classical result in graph theory due to Menger   ensures that, given two vertices $x$ and $y$ in a graph $G$, the minimum number of vertices whose deletion makes that there is no more paths between $x$ and $y$  is equal to the maximum number of internally disjoint paths between these two vertices~\cite{BondyM08}.

We can encode these two ideas with the following \QCTL\ formulas (interpreted in $x$):
\begin{align}
\Phi_k \;&\eqdef\; \exists p_1 \ldots \exists p_{k-1}. \Big( \ET_{1\leq i < k} \EX \big(\Ex (p_i \et \ET_{j\not= i} \non p_j) \:\Until\: y\big) \;\et\; \EX\: \Ex (\ET_{1\leq i <k} \non p_i) \:\Until\: y \Big)  \label{kconF1}\\
\Psi_k \;&\eqdef\; \forall^1 p_1 \ldots \forall^1 p_{k-1}. \: \EX  \: \Big( \Ex \big( \ET_{1\leq i < k} \non p_i\big) \:\Until\: y\Big) \label{kconF2}
\end{align}

$\Phi_k$ uses the labelling by the $p_i$'s to mark the internal vertices of $k$ paths between the current position and the vertex $y$. The modality $\EX$ is used to consider only the intermediate states (and not the starting state). The formula $\Psi_k$ proceeds differently: the idea is to mark exactly $k-1$ states with $p_1,\ldots,p_{k-1}$ and to verify that there still exists at least one path leading to $y$ without going through the states labelled by some $p_i$. 
By Menger's Theorem, we know that these formulas are equivalent over undirected graphs.

We interpret these formulas over Kripke structures $\calS_{n,m}$ with $n\geq m$ (see Figure~\ref{fig-kconnect}) which correspond to two kinds of grids $n\times n$ connected by $m$ edges (these edges are of the form $(q_{i,n},r_{1,i})$ or $(q_{n,i},r_{i,1})$.
The initial state is $q_{1,1}$ and when evaluating $\Phi_k$ or $\Psi_k$ we assume the state $r_{n,n}$ to be labelled by $y$.  In this context, we clearly have that $\Phi_k$ and $\Psi_k$ hold for true at $q_{1,1}$ iff $k\leq m$.

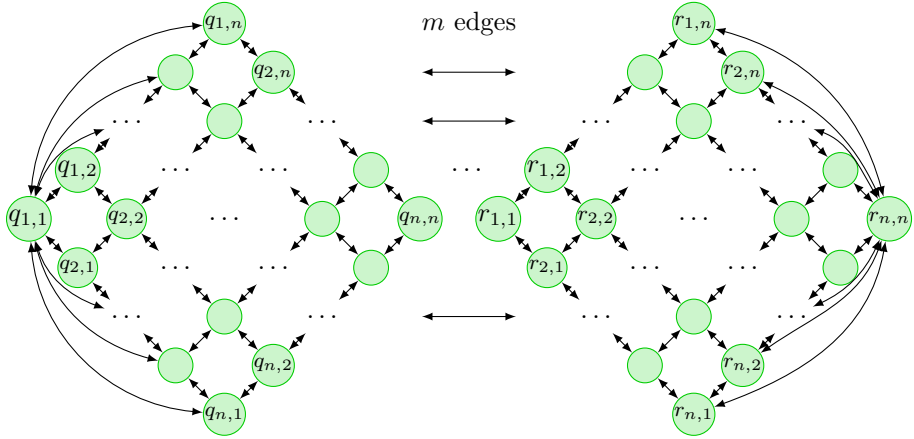
\begin{figure}[h]
\centering
\begin{tikzpicture}[inner sep=0pt,scale=1.3]
\draw (0,2) node[minirond,vert] (q00) {$q_{1,1}$} ;
\draw (0.5,2.5) node[minirond,vert] (q01) { $q_{1,2}$} ;
\draw (1,3) node[minirond,white] (q02) {\textcolor{black}{$\;\ldots\;$}} ;
\draw (1.5,3.5) node[minirond,vert] (q03) {\footnotesize {$\quad\;\;$}} ;
\draw (2,4) node[minirond,vert] (q04) {\footnotesize $q_{1,n}$} ;
\draw (0.5,1.5) node[minirond,vert] (q10) {\footnotesize $q_{2,1}$} ;
\draw (1,2) node[minirond,vert] (q11) {\footnotesize $q_{2,2}$} ;
\draw (1.5,2.5) node[minirond,white] (q12) {\textcolor{black}{$\;\ldots\;$}} ;
\draw (2,3) node[minirond,vert] (q13) {\footnotesize {$\quad\;\;$}} ;
\draw (2.5,3.5) node[minirond,vert] (q14) {\footnotesize $q_{2,n}$} ;
\draw (1,1) node[minirond,white] (q20) {\textcolor{black}{$\;\ldots\;$}} ;
\draw (1.5,1.5) node[minirond,white] (q21) {\textcolor{black}{$\;\ldots\;$}} ;
\draw (2,2) node (q22) { $\;\ldots\;$} ;
\draw (2.5,2.5) node[minirond,white] (q23) {\textcolor{black}{$\;\ldots\;$}} ;
\draw (3,3) node[minirond,white] (q24) {\textcolor{black}{$\;\ldots\;$}} ;
\draw (1.5,0.5) node[minirond,vert] (q30) {\footnotesize $\quad\;\;$} ;
\draw (2,1) node[minirond,vert] (q31) {\footnotesize $\quad\;\;$} ;
\draw (2.5,1.5) node[minirond,white] (q32) {\textcolor{black}{$\;\ldots\;$}} ;
\draw (3,2) node[minirond,vert] (q33) {\footnotesize {$\quad\;\;$}} ;
\draw (3.5,2.5) node[minirond,vert] (q34) {\footnotesize $\quad\;\;$} ;
\draw (2,0) node[minirond,vert] (q40) {\footnotesize $q_{n,1}$} ;
\draw (2.5,.5) node[minirond,vert] (q41) {\footnotesize $q_{n,2}$} ;
\draw (3,1) node[minirond,white] (q42) {\textcolor{black}{$\;\ldots\;$}} ;
\draw (3.5,1.5) node[minirond,vert] (q43) {\footnotesize {$\quad\;\;$}} ;
\draw (4,2) node[minirond,vert] (q44) {\footnotesize $q_{n,n}$} ;

\draw[>=latex,<->](q00) -- (q01) ;
\draw[>=latex,<->](q01) -- (q02) ;
\draw[>=latex,<->](q02) -- (q03) ;
\draw[>=latex,<->](q03) -- (q04) ;
\draw[>=latex,<->](q10) -- (q11) ;
\draw[>=latex,<->](q11) -- (q12) ;
\draw[>=latex,<->](q12) -- (q13) ;
\draw[>=latex,<->](q13) -- (q14) ;
\draw[>=latex,<->](q30) -- (q31) ;
\draw[>=latex,<->](q31) -- (q32) ;
\draw[>=latex,<->](q32) -- (q33) ;
\draw[>=latex,<->](q33) -- (q34) ;
\draw[>=latex,<->](q40) -- (q41) ;
\draw[>=latex,<->](q41) -- (q42) ;
\draw[>=latex,<->](q42) -- (q43) ;
\draw[>=latex,<->](q43) -- (q44) ;
\draw[>=latex,<->](q00) -- (q10) ;
\draw[>=latex,<->](q10) -- (q20) ;
\draw[>=latex,<->](q20) -- (q30) ;
\draw[>=latex,<->](q30) -- (q40) ;
\draw[>=latex,<->](q01) -- (q11) ;
\draw[>=latex,<->](q11) -- (q21) ;
\draw[>=latex,<->](q21) -- (q31) ;
\draw[>=latex,<->](q31) -- (q41) ;
\draw[>=latex,<->](q03) -- (q13) ;
\draw[>=latex,<->](q13) -- (q23) ;
\draw[>=latex,<->](q23) -- (q33) ;
\draw[>=latex,<->](q33) -- (q43) ;
\draw[>=latex,<->](q04) -- (q14) ;
\draw[>=latex,<->](q14) -- (q24) ;
\draw[>=latex,<->](q24) -- (q34) ;
\draw[>=latex,<->](q34) -- (q44) ;
\draw[>=latex,<->,out=73,in=-160] (q00) to  (q02);
\draw[>=latex,<->,out=75,in=-165] (q00) to  (q03);
\draw[>=latex,<->,out=85,in=-175] (q00) to  (q04);

\draw[>=latex,<->,out=-73,in=160] (q00) to  (q20);
\draw[>=latex,<->,out=-75,in=165] (q00) to  (q30);
\draw[>=latex,<->,out=-85,in=175] (q00) to  (q40);

\begin{scope}[shift={(4.8,0)}]
\draw (0,2) node[minirond,vert] (r00) {$r_{1,1}$} ;
\draw (0.5,2.5) node[minirond,vert] (r01) { $r_{1,2}$} ;
\draw (1,3) node[minirond,white] (r02) {\textcolor{black}{$\;\ldots\;$}} ;
\draw (1.5,3.5) node[minirond,vert] (r03) {\footnotesize {$\quad\;\;$}} ;
\draw (2,4) node[minirond,vert] (r04) {\footnotesize $r_{1,n}$} ;
\draw (0.5,1.5) node[minirond,vert] (r10) {\footnotesize $r_{2,1}$} ;
\draw (1,2) node[minirond,vert] (r11) {\footnotesize $r_{2,2}$} ;
\draw (1.5,2.5) node[minirond,white] (r12) {\textcolor{black}{$\;\ldots\;$}} ;
\draw (2,3) node[minirond,vert] (r13) {\footnotesize {$\quad\;\;$}} ;
\draw (2.5,3.5) node[minirond,vert] (r14) {\footnotesize $r_{2,n}$} ;
\draw (1,1) node[minirond,white] (r20) {\textcolor{black}{$\;\ldots\;$}} ;
\draw (1.5,1.5) node[minirond,white] (r21) {\textcolor{black}{$\;\ldots\;$}} ;
\draw (2,2) node (r22) { $\;\ldots\;$} ;
\draw (2.5,2.5) node[minirond,white] (r23) {\textcolor{black}{$\;\ldots\;$}} ;
\draw (3,3) node[minirond,white] (r24) {\textcolor{black}{$\;\ldots\;$}} ;
\draw (1.5,0.5) node[minirond,vert] (r30) {\footnotesize $\quad\;\;$} ;
\draw (2,1) node[minirond,vert] (r31) {\footnotesize $\quad\;\;$} ;
\draw (2.5,1.5) node[minirond,white] (r32) {\textcolor{black}{$\;\ldots\;$}} ;
\draw (3,2) node[minirond,vert] (r33) {\footnotesize {$\quad\;\;$}} ;
\draw (3.5,2.5) node[minirond,vert] (r34) {\footnotesize $\quad\;\;$} ;
\draw (2,0) node[minirond,vert] (r40) {\footnotesize $r_{n,1}$} ;
\draw (2.5,.5) node[minirond,vert] (r41) {\footnotesize $r_{n,2}$} ;
\draw (3,1) node[minirond,white] (r42) {\textcolor{black}{$\;\ldots\;$}} ;
\draw (3.5,1.5) node[minirond,vert] (r43) {\footnotesize {$\quad\;\;$}} ;
\draw (4,2) node[minirond,vert] (r44) {\footnotesize $r_{n,n}$} ;

\draw[>=latex,<->](r00) -- (r01) ;
\draw[>=latex,<->](r01) -- (r02) ;
\draw[>=latex,<->](r02) -- (r03) ;
\draw[>=latex,<->](r03) -- (r04) ;
\draw[>=latex,<->](r10) -- (r11) ;
\draw[>=latex,<->](r11) -- (r12) ;
\draw[>=latex,<->](r12) -- (r13) ;
\draw[>=latex,<->](r13) -- (r14) ;
\draw[>=latex,<->](r30) -- (r31) ;
\draw[>=latex,<->](r31) -- (r32) ;
\draw[>=latex,<->](r32) -- (r33) ;
\draw[>=latex,<->](r33) -- (r34) ;
\draw[>=latex,<->](r40) -- (r41) ;
\draw[>=latex,<->](r41) -- (r42) ;
\draw[>=latex,<->](r42) -- (r43) ;
\draw[>=latex,<->](r43) -- (r44) ;
\draw[>=latex,<->](r00) -- (r10) ;
\draw[>=latex,<->](r10) -- (r20) ;
\draw[>=latex,<->](r20) -- (r30) ;
\draw[>=latex,<->](r30) -- (r40) ;
\draw[>=latex,<->](r01) -- (r11) ;
\draw[>=latex,<->](r11) -- (r21) ;
\draw[>=latex,<->](r21) -- (r31) ;
\draw[>=latex,<->](r31) -- (r41) ;
\draw[>=latex,<->](r03) -- (r13) ;
\draw[>=latex,<->](r13) -- (r23) ;
\draw[>=latex,<->](r23) -- (r33) ;
\draw[>=latex,<->](r33) -- (r43) ;
\draw[>=latex,<->](r04) -- (r14) ;
\draw[>=latex,<->](r14) -- (r24) ;
\draw[>=latex,<->](r24) -- (r34) ;
\draw[>=latex,<->](r34) -- (r44) ;
\draw[>=latex,<->,out=-15,in=105] (r04) to  (r44);
\draw[>=latex,<->,out=-25,in=115] (r14) to  (r44);
\draw[>=latex,<->,out=-20,in=120] (r24) to  (r44);
\draw[>=latex,<->,out=20,in=-100] (r40) to  (r44);
\draw[>=latex,<->,out=35,in=-110] (r41) to  (r44);
\draw[>=latex,<->,out=25,in=-115] (r42) to  (r44);
\end{scope}

\draw (4,3.5) node (aux1) {} ;
\draw (5,3.5) node (aux1p) {} ;
\draw (4,3) node (aux2) {} ;
\draw (5,3) node (aux2p) {} ;
\draw (4,1) node (aux3) {} ;
\draw (5,1) node (aux3p) {} ;
\draw[>=latex,<->](aux1) -- (aux1p) ;
\draw[>=latex,<->](aux2) -- (aux2p) ;
\draw[>=latex,<->](aux3) -- (aux3p) ;
\draw (4.5,2.5) node (aux1) {\ldots} ;

\draw (4.5,4) node (aux1) {$m$ edges} ;

\end{tikzpicture}
\caption{Structure $\calS_{n,m}$ for the $k$-connectivity problem.}
\label{fig-kconnect}
\end{figure}

The first lesson of this example is that $\Phi_k$ is much more difficult to verify than $\Psi_k$: The number of temporal modalities is the main explanation. The way in which the formula is written is  is of great importance for the successful handling of even small examples.

%
%

\subsection{Nim game}

Nim game is a turn-based two-player game. A configuration is a set of heaps of objects and a boolean value indicating whose turn it is. At each turn, a player has to choose one non-empty heap and remove at least one object from it. The aim of each player is to remove the last object. Given a configuration $c$ and a Player-$J$ with $J\in\{1,2\}$, we can build a finite Kripke structure $\calS_J$, where $x_c$ is a state corresponding to the configuration $c$ ; and  use a \QCTL\ formula $\Phiwin^J$ such that $\calS, x_c \sat \Phiwin^J$ iff Player-$J$ has a wining strategy from $c$. Note that there is a simple and well-known criterion over the numbers of objects in each heap to decide who has a  winning strategy, but we consider this problem just because it is interesting to illustrate what kind of problem we can solve with  \QCTL. 

Each  configuration corresponds to a state in $\calS_J$. Every move of Player-$\bar{J}$ from a configuration $c$ to a configuration $c'$ provides a transition $(x_c,x_{c'})$ in $\calS_J$. However, a move of Player-$J$ from  $c$ to $c'$ is encoded as two transitions $x_c \fleche x_{c,c'} \fleche x_{c'}$ where $x_{c,c'}$ is  an intermediary state we   use to encode a strategy for Player-$J$ (marking $x_{c,c'}$ by an atomic proposition will correspond to Player-$J$ choosing $c'$ from  $c$). We assume that every state $x_c$ is labelled by $t_1$  if it's Player-$1$'s turn to play at $c$, and by $t_2$ otherwise. Every intermediary state $x_{c,c'}$ is labelled by $\text{int}$. 
We also label empty configurations by $w_1$ or $w_2$, depending on which player made the last move. 

Clearly, the size of $\calS$ will depend on the number of objects of each heap in the initial configuration. The formula $\Phiwin^J$ depends only on $J$:
\[
\Phiwin^J \;\eqdef\; \exists m. \Big( \AG \big(t_J \:\impl \: \EX m \big) \;\et\; \AF \big(w_J \ou (\text{int} \et \non m)\big)\Big)
\]
This formula holds true in a state corresponding to some configuration $c$ iff there  exists a labelling by $m$ such that every reachable configuration where it's Player-$J$'s turn, has a successor labelled by $m$ (thus a possible choice to do)  and every execution from the current state leads to either a winning state for Player-$J$ or a non-selected intermediary state, therefore all outcomes induced by the underlying strategy have to  verify $\F w_J$.  Note that in this example, the Kripke structure is acyclic (except the self-loops on the ending states).
There is no $\exists^1$  or $\forall^1$  operator, so  the encoding of $\mathsf{uniq}$ does not matter. Moreover the  formula is already flat, so  the methods $\MetFP$ and $\MetFPF$ are similar.

\subsection{Resources distribution}

The last example is as follows: given a Kripke structure $\calS$ and two integers $k$ and $d$, we aim at choosing at most $k$ states (called targets in the following) such that every reachable state (from the initial one) can reach a target in less than $d$ transitions. This problem can be encoded with the following \QCTL\ formula where $d$ modalities $\EX$ are nested:
\[
\Phi^{\textsf{res}}_{k,d} \;\eqdef\;  \exists^1  c_1 \ldots \exists^1 c_k . \: \AG \Big( (\!\!\OU_{1\leq i \leq k}\!\! c_i) \ou \EX \Big( 
(\!\!\OU_{1\leq i \leq k}\!\! c_i) \ou \Big( \ldots \ou \EX ( \!\!\OU_{1\leq i \leq k}\!\! c_i) \Big) \Big) \Big)
\]

For experimental  results, we consider the grid $\calK_{n,m}$ described at Figure~\ref{fig-resources} where nodes have high branching degrees. Here the interesting point  is the nesting of temporal modalities $\EX$: we will see that only the methods based on flat formulas are successful. Moreover due to the structure of the formula, the methods $\MetFBV$ and $\MetPNF$ are similar (no Until is used except the outermost $\AG$). 

\subsection{Overview of experimental results}
\label{sec-results}

\begin{table}[t]
\begin{tabular}{|l|p{2.7cm}|c|c|c|c|c|}
\hline 
 \multicolumn{2}{|l|}{problem} & size & \MetUU/\Z3 &   \MetFP/\Z3 & \MetFP/\qfm & \MetFPF/\Z3  \\
  \hline
\multicolumn{7}{|l|}{Reset property:} \\
 1 & $\calV_{10,30}\sat \Lambda_{12}$ & 301 & 76+0.7 &  38+2 & 40+7 & 42+2  \\
 2 & $\calV_{15,100}\sat \Lambda_{16}$ & 1501 & X &  X & X & X  \\
 3 & $\calV_{6,10}\not\sat \Lambda_{5}$ & 61 & \textbf{0.4+25} &  \textbf{0.5+42} & 0.5+122 & X   \\
 \hline
 \multicolumn{7}{|l|}{$k$-Connectivity:} \\
 4 & $\calS_{10,5}\sat \Psi_4$ & 200  & X & 
 \textbf{2+0.6} &   {3+2} & X   \\
 5 & $\calS_{15,5}\sat \Psi_4$  & 450 & X  &  15+4 & 16+5 & X \\
 6 & $\calS_{15,7}\sat \Psi_6$ & 450 & X &  22+112 & \textbf{25+56}& X  \\
 7 & $\calS_{30,6}\sat \Psi_4$ & 1800 & X & 188+126 & 196+100 & X \\
 8 & $\calS_{10,4}\not\sat \Psi_5$  & 200 & X  &  \textbf{3+0.7} &  \textbf{3+1} & 84+X  \\
  \hline
   \multicolumn{7}{|l|}{Nim game:} \\
   9 & $[3,4,5]\sat \Phiwin^1$ & 96  & 40+0.2 &  \textbf{0.1+0.1}
   & 0.1+X   & (\MetFP)   \\ 
  10 & $[2,3,4,4]\sat \Phiwin^1$ & 124  & 157+0.3 & \textbf{ 0.1+0.1} & 0.1+X   & (\MetFP)  \\ 
  11 & $[3,4,5,6]\sat \Phiwin^1$ &  330 & X  &  \textbf{0.2+0.2} & 0.2+X   & (\MetFP)  \\ 
  12 & $[2,4,8,14]\not\sat \Phiwin^1$ &  1556 & X &  \textbf{3+33} & 3+X   & (\MetFP)   \\ 
  \hline
  \multicolumn{7}{|l|}{Resources distribution:} \\
  13 & $\calL_{10,10}\sat \Phi^{\textsf{res}}_{8,6}  $ & 100  &  X & 
X  & X   &  \textbf{14+0.3}   \\ 
   14 & $\calL_{12,12}\sat \Phi^{\textsf{res}}_{8,6}  $ & 144   & X &  X
  & X   &  \textbf{30+0.7}   \\ 
  15 & $\calL_{12,12}\sat \Phi^{\textsf{res}}_{6,8}  $ & 144  & X &  X
  & X   &  \textbf{28+0.6}   \\ 
  16 & $\calL_{20,20}\sat \Phi^{\textsf{res}}_{6,8}  $ & 400 & X &  X
  & X   &  \textbf{153+3}   \\ 

  \hline
\end{tabular}
\caption{Overview of experimental results (1).}
\label{tabres1}
\end{table}

\begin{table}[ht]
\begin{tabular}{|l|c|c|c|c|c|}
\hline
  $\sharp$ pb &   \MetPNF/\Z3 & \MetPNF/\qfm & \MetPNF/\cqesto & \MetPNF/\qfun &  \MetFBV/\cqesto \\
  \hline
  \multicolumn{6}{|l|}{Reset property:} \\
 1 &    1+4 & \textbf{1+1} & \textbf{1+0.2}  &  \textbf{1+0.3} & 110+12 \\
  2 &   7+X &  \textbf{7+9} & \textbf{8+2} & \textbf{8+2}  &  X \\
  3 &  \textbf{0.4+25} &  0.7+94  & 0.1+156 & \textbf{0.2+30}  & 3+X  \\
 \hline
  \multicolumn{6}{|l|}{$k$-Connectivity:} \\
  4 &  0.2+5 & 0.2+46 & \textbf{0.2+0.1} & 0.2+5 & 48+X  \\
   5  &  0.4+16 & 0.4+X &  \textbf{0.4+6} & 0.4+487 & 507+X \\
   6 &  0.7+X & 0.7+X & 0.7+X & 0.7+X & X \\
   7 &   2+587 & \textbf{2+9} & \textbf{4+4.3}  & 4+X & X \\
   8  &  \textbf{0.2+2} & {0.2+10} & 0.2+19 & {0.2+8} & 131+X \\
  \hline
     \multicolumn{6}{|l|}{Nim game:} \\
    9 &  0.1+123  & \textbf{0.1+1} & \textbf{0.1+1}  & 0.1+31   & 212+28   \\ 
   10 &  0.1+49  & 0.1+18 & \textbf{0.1+1} & 0.1+X   & 317+50   \\ 
   11 &  0.3+X & 0.3+X  & 0.3+X & 0.3+X   &  X   \\ 
   12 &  4+X & 4+X & 4+X & 4+X   & X   \\ 
  \hline
    \multicolumn{6}{|l|}{Resources distribution:} \\
   13 & \textbf{16+0.4} & \textbf{17+0.3} & \textbf{18+0.2}     & \textbf{17+0.2}   &  (\MetPNF)  \\ 
    14 &  \textbf{34+0.6}  & \textbf{35+0.6}  & \textbf{35+0.3}
 &  \textbf{36+0.4}  &  (\MetPNF)   \\ 
   15 &  \textbf{39+0.7}  & \textbf{35+0.5}  & \textbf{39+0.3}
 &   \textbf{35+0.3}  &  (\MetPNF)   \\ 
     16 &  \textbf{207+3} & \textbf{190+3} & \textbf{180+2} 
  & \textbf{179+3}   &   (\MetPNF) \\ 
  \hline
\end{tabular}
\label{tabres2}
\caption{Overview of experimental results (2).}
\end{table}

The main experimental results are given in Tables~\ref{tabres1} and~4. 
We distinguish the time to build the \QBF \ formula and the solver's part (a timeout is set to 600 seconds). For example, for  $\calV_{10,30}\sat \Lambda_{12}$, the result is $38+2$, it means that the construction of the \QBF \ formula took 38s and the \QBF\ solver needed 2s to decide whether the formula was valid or not. A result of the form $7+X$ means that the  formula was built in $7s$ but the solver didn't give the result in less than 600s. And a result of the form $X$ means that the \QBF\ formula was not yet built after 600s. Moreover $(A)$ means that this reduction is similar to the reduction $A$ (\ie\ the two preprocessing steps provide the same formula). Best results  for a case are in bold.

The first remark is that we can  successfully apply our techniques and get answers for \QCTL\ verification problems. In the tables, we just  selected the most significant results.  The encoding of $\exists^1$ and $\forall^1$ is always done with bit vectors: indeed the two other encodings presented in section~\ref{sec-ea1} are never the best solution on the examples we consider.

The results show that the two most interesting methods seem to be the \MetFP\ reduction (associated with \Z3) and the \MetPNF\ reduction (associated with \cqesto). As soon as the formula contains nested temporal modalities the \MetPNF \ reduction is better. For the Nim game, the reduction \MetFP\ is much more efficient than any other reduction and give the solution for structures with more than 1500 states. For the resource distribution (with a high branching degree of the structure and a high temporal nesting in the formula), the flattening is mandatory (\MetPNF\ or \MetFPF). 

The results also depend  on the choice of the solver: a \QBF\ formula  can be solved very easily by a solver but requires a very long execution time for others.

\section{Conclusion}

We have presented several reductions from \QCTL\ model-checking to \QBF. This provides a first tool for \QCTL\ model-checking with the structure semantics. These first results are rather interesting and encouraging. 
We have seen the importance of writing "good"  \QCTL\ formulas for which the solver will be able to provide a result (this problem already exists for classical temporal logics, but it is more significant here due to the complexity induced by the quantifications). 
The examples also show that there is no "one best strategy" and no "one best solver": the best choices depend on the structure of the considered formula, the structure of the model. Two reductions seem to be the most interesting: \MetFP\ for simple \QCTL\ formulas (with few nesting of temporal modalities) and \MetPNF\ for more complex formulas. Considering several solvers is  an important point: the heuristics used by the solvers may affect significantly the execution times.  
In the future, we plan to continue to work on reduction strategies for \QCTL. Considering special methods for timed modalities of the form $\EF_{< d} \: \vfi$ (\ie\ "a state satisfying $\vfi$ is reachable in less than $d$ transitions") would be very useful in practice and  the reduction  \MetFBV \ could be useful for it. Considering other solvers (based on prenex CNF formulas) would be also an interesting work. 
Finally  we also plan to consider other logics for which such  reductions to \QBF~ are possible to get decision procedures.

\paragraph{Acknowledgement} 
We would like to strongly thank Mikol{\'{a}}s Janota for his help and advice about the \QBF\ solvers, as well as Yann Regis-Gianas who helped us with the experiments part, and the anonymous reviewers for their helpful suggestions to improve the paper.

\bibliography{bibic}

\begin{thebibliography}{10}

\bibitem{AFH15}
Carlos Areces, Raul Fervari, and Guillaume Hoffmann.
\newblock Relation-changing modal operators.
\newblock {\em Logic Journal of the {IGPL}}, 23(4):601--627, 2015.

\bibitem{ABG18}
Guillaume Aucher, Johan van Benthem, and Davide Grossi.
\newblock Modal logics of sabotage revisited.
\newblock {\em J. Log. Comput.}, 28(2):269--303, 2018.

\bibitem{Biere-SAT-Race-2019-solvers}
Armin Biere.
\newblock {CaDiCaL} at the {SAT} {Race} 2019.
\newblock In Marijn Heule, Matti J{\"a}rvisalo, and Martin Suda, editors, {\em
  Proc.~of {SAT Race} 2019 -- Solver and Benchmark Descriptions}, volume
  B-2019-1 of {\em Department of Computer Science Series of Publications B},
  pages 8--9. University of Helsinki, 2019.

\bibitem{BiereCCSZ03}
Armin Biere, Alessandro Cimatti, Edmund~M. Clarke, Ofer Strichman, and Yunshan
  Zhu.
\newblock Bounded model checking.
\newblock {\em Advances in Computers}, 58:117--148, 2003.

\bibitem{BiereCCZ99}
Armin Biere, Alessandro Cimatti, Edmund~M. Clarke, and Yunshan Zhu.
\newblock Symbolic model checking without bdds.
\newblock In {\em Tools and Algorithms for Construction and Analysis of
  Systems, 5th International Conference, {TACAS} '99, Amsterdam, The
  Netherlands, March 22-28, 1999, Proceedings}, volume 1579 of {\em Lecture
  Notes in Computer Science}, pages 193--207. Springer, 1999.

\bibitem{BondyM08}
J.~Adrian Bondy and Uppaluri S.~R. Murty.
\newblock {\em Graph Theory}.
\newblock Graduate Texts in Mathematics. Springer, 2008.

\bibitem{CE82}
Edmund~M. Clarke and E.~Allen Emerson.
\newblock Design and synthesis of synchronization skeletons using
  branching-time temporal logic.
\newblock In Dexter~C. Kozen, editor, {\em {P}roceedings of the 3rd {W}orkshop
  on {L}ogics of {P}rograms ({LOP}'81)}, volume 131 of {\em Lecture Notes in
  Computer Science}, pages 52--71. Springer-Verlag, 1982.

\bibitem{qfm}
Simon Cooksey, Sarah Harris, Mark Batty, Radu Grigore, and Mikolas Janota.
\newblock Pridemm: Second order model checking for memory consistency models.
\newblock In {\em Pre-proceedings of TAPAS 2019, 10th Workshop on Tools for
  Automatic Program Analysis}, pages 7--26, 2019.

\bibitem{DLM12}
Arnaud Da{~}Costa, Fran{\c c}ois Laroussinie, and Nicolas Markey.
\newblock Quantified~{CTL}: Expressiveness and model checking.
\newblock In Maciej Koutny and Irek Ulidowski, editors, {\em {P}roceedings of
  the 23rd {I}nternational {C}onference on {C}oncurrency {T}heory
  ({CONCUR}'12)}, volume 7454 of {\em Lecture Notes in Computer Science}, pages
  177--192. Springer-Verlag, September 2012.

\bibitem{MouraB08}
Leonardo~Mendon{\c{c}}a de~Moura and Nikolaj Bj{\o}rner.
\newblock {Z3:} an efficient {SMT} solver.
\newblock In {\em Tools and Algorithms for the Construction and Analysis of
  Systems, 14th International Conference, {TACAS} 2008, Budapest, Hungary,
  2008. Proceedings}, volume 4963 of {\em Lecture Notes in Computer Science},
  pages 337--340. Springer, 2008.

\bibitem{DershowitzHK05}
Nachum Dershowitz, Ziyad Hanna, and Jacob Katz.
\newblock Bounded model checking with {QBF}.
\newblock In {\em Theory and Applications of Satisfiability Testing, 8th
  International Conference, {SAT} 2005, St. Andrews, UK, June 19-23, 2005,
  Proceedings}, volume 3569 of {\em Lecture Notes in Computer Science}, pages
  408--414. Springer, 2005.

\bibitem{EenS03}
Niklas E{\'{e}}n and Niklas S{\"{o}}rensson.
\newblock An extensible sat-solver.
\newblock In {\em Theory and Applications of Satisfiability Testing, 6th
  International Conference, {SAT} 2003. Santa Margherita Ligure, Italy, May
  5-8, 2003 Selected Revised Papers}, volume 2919 of {\em Lecture Notes in
  Computer Science}, pages 502--518. Springer, 2003.

\bibitem{Fre01}
Tim French.
\newblock Decidability of quantified propositional branching time logics.
\newblock In Markus Stumptner, Dan Corbett, and Mike Brooks, editors, {\em
  {P}roceedings of the 14th {A}ustralian {J}oint {C}onference on {A}rtificial
  {I}ntelligence ({AJCAI}'01)}, volume 2256 of {\em Lecture Notes in Computer
  Science}, pages 165--176. Springer-Verlag, December 2001.

\bibitem{Fre03}
Tim French.
\newblock Quantified propositional temporal logic with repeating states.
\newblock In {\em {P}roceedings of the 10th {I}nternational {S}ymposium on
  {T}emporal {R}epresentation and {R}easoning and of the 4th {I}nternational
  {C}onference on {T}emporal {L}ogic ({TIME}-{ICTL}'03)}, pages 155--165. IEEE
  Comp. Soc. Press, July 2003.

\bibitem{cqesto}
Mikol{\'{a}}s Janota.
\newblock Circuit-based search space pruning in {QBF}.
\newblock In {\em Theory and Applications of Satisfiability Testing - {SAT}
  2018 - 21st International Conference, {SAT} 2018, Oxford, UK, July 9-12,
  2018, Proceedings}, volume 10929 of {\em Lecture Notes in Computer Science},
  pages 187--198. Springer, 2018.

\bibitem{Kup95a}
Orna Kupferman.
\newblock Augmenting branching temporal logics with existential quantification
  over atomic propositions.
\newblock In Pierre Wolper, editor, {\em {P}roceedings of the 7th
  {I}nternational {C}onference on {C}omputer {A}ided {V}erification
  ({CAV}'95)}, volume 939 of {\em Lecture Notes in Computer Science}, pages
  325--338. Springer-Verlag, July 1995.

\bibitem{LM14}
Fran{\c{c}}ois Laroussinie and Nicolas Markey.
\newblock Quantified {CTL:} expressiveness and complexity.
\newblock {\em Logical Methods in Computer Science}, 10(4), 2014.

\bibitem{LaroussinieM15}
Fran{\c{c}}ois Laroussinie and Nicolas Markey.
\newblock Augmenting {ATL} with strategy contexts.
\newblock {\em Inf. Comput.}, 245:98--123, 2015.

\bibitem{LMS15}
Fran{\c{c}}ois Laroussinie, Nicolas Markey, and Arnaud Sangnier.
\newblock Atlsc with partial observation.
\newblock In {\em Proceedings Sixth International Symposium on Games, Automata,
  Logics and Formal Verification, GandALF 2015, Genoa, Italy, 21-22nd September
  2015}, volume 193 of {\em {EPTCS}}, pages 43--57, 2015.

\bibitem{LR03}
Christof L{\"{o}}ding and Philipp Rohde.
\newblock Model checking and satisfiability for sabotage modal logic.
\newblock In {\em {FST} {TCS} 2003: Foundations of Software Technology and
  Theoretical Computer Science, 23rd Conference, Mumbai, India, December 15-17,
  2003, Proceedings}, volume 2914 of {\em Lecture Notes in Computer Science},
  pages 302--313. Springer, 2003.

\bibitem{LR03b}
Christof L{\"{o}}ding and Philipp Rohde.
\newblock Solving the sabotage game is pspace-hard.
\newblock In {\em Mathematical Foundations of Computer Science 2003, 28th
  International Symposium, {MFCS} 2003, Bratislava, Slovakia, August 25-29,
  2003, Proceedings}, volume 2747 of {\em Lecture Notes in Computer Science},
  pages 531--540. Springer, 2003.

\bibitem{McMillan02}
Kenneth~L. McMillan.
\newblock Applying {SAT} methods in unbounded symbolic model checking.
\newblock In {\em Computer Aided Verification, 14th International Conference,
  {CAV} 2002,Copenhagen, Denmark, July 27-31, 2002, Proceedings}, volume 2404
  of {\em Lecture Notes in Computer Science}, pages 250--264. Springer, 2002.

\bibitem{PBDDC02}
Anindya~C. Patthak, Indrajit Bhattacharya, Anirban Dasgupta, Pallab Dasgupta,
  and P.~P. Chakrabarti.
\newblock Quantified computation tree logic.
\newblock {\em Information Processing Letters}, 82(3):123--129, 2002.

\bibitem{Pnu77}
Amir Pnueli.
\newblock The temporal logic of programs.
\newblock In {\em {P}roceedings of the 18th {A}nnual {S}ymposium on
  {F}oundations of {C}omputer {S}cience ({FOCS}'77)}, pages 46--57. IEEE Comp.
  Soc. Press, October-November 1977.

\bibitem{QS82a}
Jean-Pierre Queille and Joseph Sifakis.
\newblock Specification and verification of concurrent systems in {CESAR}.
\newblock In Mariangiola Dezani{-}Ciancaglini and Ugo Montanari, editors, {\em
  {P}roceedings of the 5th {I}nternational {S}ymposium on {P}rogramming
  ({SOP}'82)}, volume 137 of {\em Lecture Notes in Computer Science}, pages
  337--351. Springer-Verlag, April 1982.

\bibitem{qfun}
Ricardo~Joel Silva and Mikolas Janota.
\newblock Machine learning of strategies in qbf solving.
\newblock In {\em Pre-proceedings of 26th RCRA International Workshop on
  "Experimental Evaluation of Algorithms for Solving Problems with
  Combinatorial Explosion"}, 2019.

\bibitem{sis83}
A.~Prasad Sistla.
\newblock {\em Theoretical Issues in the Design and Verification of Distributed
  Systems}.
\newblock PhD thesis, Harvard University, Cambridge, Massachussets, USA, 1983.

\bibitem{stirling2001}
Colin Stirling.
\newblock {\em Modal and Temporal Properties of Processes}.
\newblock Springer (Texts in Computer Science), 2001.

\bibitem{Benthem05}
Johan van Benthem.
\newblock An essay on sabotage and obstruction.
\newblock In {\em Mechanizing Mathematical Reasoning, Essays in Honor of
  J{\"{o}}rg H. Siekmann on the Occasion of His 60th Birthday}, volume 2605 of
  {\em Lecture Notes in Computer Science}, pages 268--276. Springer, 2005.

\end{thebibliography}

\appendix 


\section{Proof of Theorem~\ref{theo-UU}}

\label{app-uu}

\begin{proof}
We assume $\Phi$ to be fixed and we prove the property by structural induction over the subformula $\vfi$. Boolean operators are omitted. 
\begin{itemize}
\item $\vfi = p$: if  $\calK,x\sat_\varepsilon p$, then either $p$ is a quantified proposition and  $x$ belongs to $\epsilon(p)$ and thus $v_\varepsilon \sat p^x$ by definition of $v_\varepsilon$, or $p$ belongs to $\ell(x)$. In  both cases we have $v_\varepsilon \sat \widehat{p}^{x,\dom(\varepsilon)}$. The converse is similar. 

\item $\vfi=\exists p. \psi$. We have $\calK,x\sat_\varepsilon \exists p. \psi$ iff there exists  $V'\subseteq V$ s.t.\ $\calK,x \sat_{\varepsilon[p\mapsto V']} \psi$, iff (by i.h.) there exists $V'\subseteq V$ s.t.\ $v_{\varepsilon[p\mapsto V']} \sat \widehat{\psi}^{x,\dom(\varepsilon)\cup\{p\}}$ which is equivalent to $v_{\varepsilon} \sat \exists p^{v_1}\ldots p^{v_n}.$ $\widehat{\psi}^{x,\dom(\varepsilon)\cup\{p\}}$ (by definition of    $v_\varepsilon$). 

\item $\vfi=\EX \psi$:   $\calK,x\sat_\varepsilon \EX \psi$ iff there exists $(x,x')\in E$ s.t.\ $\calK,x ' \sat_\varepsilon \psi$, iff (by i.h.)  there exists $(x,x')\in E$ s.t.\  $v_\varepsilon \sat \widehat{\psi}^{x',\dom(\varepsilon)}$ which is equivalent to $v_\varepsilon \sat \widehat{\EX \psi}^{x,\dom(\varepsilon)}$. 

\item $\vfi=\AX \psi$: Similar to $\EX$ with a conjunction to ensure that \emph{all} successors satisfy $\psi$. 
\item $\vfi=\EF \psi$ or $\vfi=\AG \psi$: Similar to $\EX$ or $\AX$ except that we consider any reachable state $x'$ instead of immediate successors (thus we use the reflexive and transitive closure $E^*$ of $E$). 

\item $\vfi=\Ex \psi_1 \Until \psi_2$: The definition of $\widehat{\vfi}^{x,\dom(\varepsilon)}$ corresponds to a finite unfolding of the expansion law that characterizes the $\EU$\ modality. Assume  $\calK,x\sat_\varepsilon \vfi$. There exists a  path $\rho \in 
\Path^\omega_\calK(x)$ and a position $i\geq 0$ s.t.\ $\rho(i) \sat_\varepsilon  \psi_2$ and $\rho(k) \sat_\varepsilon  \psi_1$ for any $0 \leq k < i$. The finite prefix $x=\rho(0)\cdots\rho(k)$ can be assumed to be simple, and then $k<|V|$. By using i.h., we get 
$v_\varepsilon \sat \widehat{\psi_2}^{\rho(i),\dom(\varepsilon)}$, and $v_\varepsilon \sat \widehat{\psi_1}^{\rho(k),\dom(\varepsilon)}$ for any $0\leq k <i$. From this point, the reader can easily verify by induction (starting at $i$, down to 0) that $v_\varepsilon\sat\overline{\Ex \psi_1\Until\psi_2}^{\rho(k), dom(\varepsilon), \{\rho(j)|j\leq k\}}$ for all $k\leq i$. This makes $\widehat{\vfi}^{x,\dom(\varepsilon)}$ to be satisfied by $v_\varepsilon$.

Conversely, assume $v_\varepsilon \sat \widehat{\vfi}^{x,\dom(\varepsilon)}$, \ie\ 
$v_\varepsilon \sat \overline{\Ex \psi_1 \Until \psi_2}^{x,\dom(\varepsilon),\{x\}}$. Given the definition, there exists a sequence of states $x_0, \ldots, x_i$ s.t.\ (1) $x_0=x$, (2) $v_\varepsilon \sat  \widehat{\psi_2}^{x_i,\dom(\varepsilon)}$ and (3) for any $0 \leq j <i$ we have $v_\varepsilon \sat \widehat{\psi_1}^{x_j,\dom(\varepsilon)}$, $(x_j,x_{j+1})\in E$ and $x_{j+1} \not\in \{x_0,\ldots,x_j\}$. And by i.h., we can deduce that $x_0 \ldots x_i$ is a path in $\calK$ satisfying $\psi_1 \Until \psi_2$.

\item $\vfi=\All \psi_1 \Until \psi_2$: this case is similar to the previous one, except that we have to consider loops.
Assume $\calK,x\sat_\varepsilon \vfi$. Then any path issued from $x$  satisfies $\psi_1 \Until \psi_2$:  
either there is a simple prefix witnessing $\psi_1 \Until \psi_2$ (and ending with a state satisfying $\psi_2$), or there is a loop from some point. In the   latter case, one of the state in the loop has to verify $\psi_2$. In both cases, the definition of   $\widehat{\vfi}^{x,\dom(\varepsilon)}$ gives the result. 

\end{itemize}
\hfill$\Box$
\end{proof}

\section{Proof of Proposition~\ref{prop-flat2}}
\label{app-meth5}

\begin{proof}
Consider w.l.o.g.\ a \QCTL\ formula   $\Phi$ in prenex normal form and NNF. We can define $\Phi_0$ and the basic formulas   $\theta_i$s    as in Proposition~\ref{prop-ctl-flat}, except that $\Phi_0$ contains only $\EX$,  $\AX$, $\EF$ or $\AG$ modalities, and every other modality is associated with  some quantified proposition $\kappa$ and a subformula $\AG(\ldots)$ in the main conjunction of $\Psi$. Every $\theta_i$ starts with  
a modality in  $S_{tmod}$. Let $\vfi_i$ be the original $\Phi$-subformula   associated with $\theta_i$. 
Note that   $\Psi$ is in NNF, and  $\kappa_i$ occurs only once in $\Psi$ in the scope of a negation, and it happens in the subformula $\AG(\kappa_i \impl \theta_i)$. We now have to show that $\Phi$ is equivalent to $\Psi$. Consider the formula $\widetilde{\Psi}$ where every $\impl$ is replaced by $\equivaut$: by following the same arguments of Proposition~\ref{prop-ctl-flat}, we clearly have $\Phi \equiv \widetilde{\Psi}$, and $\widetilde{\Psi} \impl \Psi$. It remains to prove the opposite direction.

To prove $\Psi \impl  \widetilde{\Psi}$,  it is sufficient to show that this is true for the empty $\calQ$ (as equivalence is substitutive). 
Assume $\calK, x \sat_\varepsilon \Psi$. Then there exists an environment $\varepsilon'$  from  $\{\kappa_1,\ldots,\kappa_m\}$ to $2^V$ such that ${\displaystyle \calK, x \sat_{\varepsilon\compo\varepsilon'} \Phi_0 \et \ET_i \AG (\kappa_i \impl \theta_i)}$.
Now we have:
\[
\forall i,\quad \calK, x  \sat_{\varepsilon \compo\varepsilon'} \theta_i \;\;\impl\;\; \calK, x  \sat_{\varepsilon}  \vfi_i  \quad \mbox{and} \quad \calK, x  \sat_{\varepsilon \compo\varepsilon'}  \Phi_0 \;\;\impl\;\; \calK, x  \sat_{\varepsilon\compo\varepsilon'}  \widetilde{\Phi_0}
\]
Indeed, assume that it is not true and $\calK, x  \sat_{\varepsilon \compo\varepsilon'} \theta_i$ and  $\calK, x  \not\sat_{\varepsilon}  \vfi_i$.  Consider such a  formula $\vfi_i$ with the smallest temporal height. The only atomic propositions $\kappa_j$ occurring in $\theta_i$ are then associated with some $\theta_j$ and $\vfi_j$ which verify the property and thus any state satisfying such a $\theta_j$, also satisfies $\vfi_j$. Therefore any state labelled by such a $\kappa_j$ is correctly labelled (and satisfies $\vfi_j$). And the states that are not labelled by  $\kappa_j$ cannot make $\theta_i$ to be wrongly evaluated to true (because $\kappa_j$ is not in the scope of a negation). Therefore $\vfi_i$ holds true at $x$. The same holds for $\Phi_0$ and $\widetilde{\Phi_0}$.
As a direct consequence, we have  $\Psi \impl  \widetilde{\Psi}$.
\hfill$\Box$
\end{proof}

\section{Proposition for Section~\ref{sec-ea1}}
\label{app-ea1}

NB: we assume that every quantifier introduces a fresh atomic proposition. 

\begin{proposition}
\label{prop-ea1}
Let $\calQ$, $\calQ'$ and $\calQ''$ be three blocks of quantifiers. We have the following equivalence:
\begin{multline}
\calQ \:\exists\: p_1 \:\calQ'\: \forall p_2 \:\calQ'' . \Big( \uniq(p_1) \:\et\: \big( \uniq(p_2) \:\impl\: \Phi \big)\Big) \quad \equiv \\ 
\calQ \:\exists\: p_1 \:\calQ'\: \forall p_2 \:\calQ'' . \Big( \uniq(p_2) \:\impl\: \big( \uniq(p_1) \:\et\: \Phi \big)\Big)
\end{multline}

\end{proposition}
\begin{proof}
\begin{itemize}
\item $(1)\impl(2)$: there exists a labelling for $p_1$ s.t.\ $\uniq(p_1)$ is true and for any $p_2$-labelling, if $\uniq(p_2)$ holds true, then  we have $\Phi$. Now choose the same  $p_1$-labelling to evaluate the right-hand side formula, we know that $\uniq(p_1)$ is true, and moreover for every $p_2$-labelling satisfying $\uniq(p_2)$, $\phi$ is satisfied. This provides the result. 
\item $(2)\impl(1)$: Choose the $p_1$-labelling. We know that there exists a $p_2$-labelling such that $\uniq(p_2)$ is true (\emph{e.g.} by labelling only the current state with $p_2$), this implies that we have $\uniq(p_1)$ and $\Phi$. And therefore, for every $p_2$-labelling satisfying $\uniq(p_2)$, we have $\Phi$.
\end{itemize}
\hfill$\Box$
\end{proof}

\end{document}